\documentclass[a4paper,UKenglish,cleveref, autoref, thm-restate]{lipics-v2021}

\nolinenumbers

\hideLIPIcs  

\title{Clustering Permutations under the Ulam Metric:\\ A Parameterized Complexity Study} 

\titlerunning{Clustering Permutations under the Ulam Metric: A Parameterized Complexity Study} 

\author{Tian Bai}{University of Bergen, Bergen, Norway}{Tian.Bai@uib.no}{https://orcid.org/0000-0003-1669-285X}{Funded by the European Union, GA{\small \#}101126560; Bergen research and training program for future AI leaders across the disciplines, LEAD AI and by the Trond Mohn forskningsstiftelse (grant no. TMS2023TMT01).}
\author{Fedor V. Fomin}{University of Bergen, Bergen, Norway}{Fedor.Fomin@uib.no}{https://orcid.org/0000-0003-1955-4612}{Supported by the European Research Council (ERC) under the European Union's Horizon 2020 research and innovation programme (NewPC grant agreement No.  101199930).\\\begin{minipage}{0.2\textwidth}\includegraphics[width=0.9\textwidth]{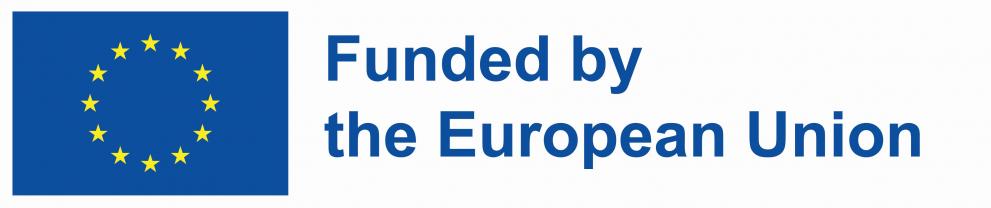}\end{minipage} }
\author{Petr A. Golovach}{University of Bergen, Bergen, Norway}{petr.golovach@uib.no}{https://orcid.org/0000-0002-2619-2990}{Supported by the Research Council of Norway under the BWCA (grant no.~314528) and  Extreme-Algorithms (grant no~355137) projects.}
\author{Yash Hiren More}{University of Bergen, Bergen, Norway}{yash.h.more@uib.no}{https://orcid.org/0000-0003-8651-6686}{Supported by the Trond Mohn forskningsstiftelse (grant no. TMS2023TMT01).}
\author{Simon Wietheger}{TU Wien, Vienna, Austria}{swietheger@ac.tuwien.ac.at}{https://orcid.org/0000-0002-0734-0708}{Supported by the Austrian Science Foundation (FWF, project 10.55776/Y1329).}

\authorrunning{T. Bai, F. V. Fomin, P. A. Golovach, Y. H. More, S. Wietheger} 

\Copyright{Tian Bai, Fedor V. Fomin, Petr. A. Golovach, Yash Hire More, Simon Wietheger} 



\category{} 

\relatedversion{} 

\acknowledgements{}

\EventEditors{Sayan Bhattacharya, Danupon Nanongkai, Michael Benedikt, and Gabriele Puppis}
\EventNoEds{4}
\EventLongTitle{53rd International Colloquium on Automata, Languages, and Programming (ICALP 2026)}
\EventShortTitle{ICALP 2026}
\EventAcronym{ICALP}
\EventYear{2026}
\EventDate{July 7--10, 2026}
\EventLocation{Royal Holloway, University of London, Egham, United Kingdom}
\EventLogo{}
\SeriesVolume{374}
\ArticleNo{57}

\usepackage{graphicx} 
\usepackage{booktabs,array,amsmath,amsthm,amsfonts,mathtools,xspace,amssymb,utfsym}
\usepackage{complexity}
\usepackage{color,boxedminipage,framed}

\usepackage{footnote,nicefrac,todonotes,xifthen,tabularx}
\usepackage[ruled,linesnumbered]{algorithm2e}

\usepackage{tikz,tcolorbox} 
\usetikzlibrary{arrows.meta}
\usetikzlibrary{decorations.pathreplacing}
\usetikzlibrary{calc,math,patterns,shapes,backgrounds}
\usepackage{algorithm2e}

\usepackage[dvipsnames]{xcolor}

\newtheorem{reduction}{Reduction}

\newlength{\RoundedBoxWidth}
\newsavebox{\GrayRoundedBox}
\newenvironment{GrayBox}[1]%
  {
    \setlength{\RoundedBoxWidth}{.96\textwidth}
    \def\boxheading{#1}
    \begin{lrbox}{\GrayRoundedBox}
      \begin{minipage}{\RoundedBoxWidth}%
  }
  {
      \end{minipage}
    \end{lrbox}
    \begin{center}
      \begin{tikzpicture}
        \node (Text) [
          fill = none,
          rounded corners = 2ex,
          inner sep = 2ex,
          text width = \RoundedBoxWidth
        ]{\usebox{\GrayRoundedBox}};

        \coordinate (X) at (current bounding box.north west);
        \node (Title) [draw = none,
          rectangle,
          inner sep = 3pt,
          anchor = north west,
          fill = none]
        at ($(X) + (6pt, .75em)$) {\boxheading};
        
        \draw[black!20, line width=0.4pt, rounded corners=2ex] 
        let 
          \p1 = (Text.north west),
          \p2 = (Text.north east),
          \p3 = (Text.south east),
          \p4 = (Text.south west),
          \p6 = (Title.east),
        in
          (\x6, \y1) -- (\p2) -- (\p3) -- (\p4) -- (\x1,  \y1 - 2ex);

        \draw[black!20, line width=0.4pt, rounded corners=2ex] 
        let 
          \p1 = (Text.north west),
          \p5 = (Title.west),
        in
          (\x1,  \y1 - 2ex)arc[start angle = 180, end angle = {acos((\x5 - (\x1 + 2ex))/2ex)}, radius = 2ex];

      \end{tikzpicture}
    \end{center}%
  }

\newenvironment{defproblemx}[2][]{\noindent\ignorespaces%
                                \FrameSep=6pt%
                                \parindent=0pt%
                \vspace*{-1.5em}
                \ifthenelse{\isempty{#1}}{%
                  \begin{GrayBox}{\textsc{#2}}%
                }{%
                  \begin{GrayBox}{\textsc{#2}  parameterized by~{#1}}%
                }
                \begin{tabular*}{\textwidth}{@{\hspace{.1em}} >{\itshape} p{.85cm} p{0.9\textwidth} @{}}%
            }{
    \vspace{-5mm}
                \end{tabular*}%
                \end{GrayBox}%
                \ignorespacesafterend
            }

\newcommand{\defproblema}[3]{
  \begin{defproblemx}{#1}
    Input:  & #2 \\
    Task: & #3
  \end{defproblemx}
}%

\newcommand{\ucenterfull}{\textsc{Ulam Metric $k$-Center}\xspace}
\newcommand{\umedianfull}{\textsc{Ulam Metric $k$-Median}\xspace}
\newcommand{\ucenterone}{\textsc{Ulam Metric $1$-Center}\xspace}

\newcommand{\vcfull}{\textsc{Vertex Cover}\xspace}
\newcommand{\cstringfull}{\textsc{Closest String}\xspace}

\newcommand{\yes}{{YES}}
\newcommand{\no}{{NO}}
\newcommand{\LCS}{\mathrm{LCS}}
\newcommand{\yesinstance}{\yes-instance\xspace}
\newcommand{\noinstance}{\no-instance\xspace}
\newcommand{\yesinstances}{\yes-instances\xspace}
\newcommand{\noinstances}{\no-instances\xspace}
\newcommand{\ulamdistance}{Ulam distance\xspace}

\newcommand{\dist}[0]{\mathrm{dist}_{U}}
\newcommand{\distH}[0]{\mathrm{dist}_{H}}
\newcommand{\base}[0]{\mathrm{base}}

\newcommand{\highlight}[1]{%
  \colorbox{black!20}{$\displaystyle#1$}}
\newcommand{\HIGHlight}[1]{%
  \colorbox{black!40}{$\displaystyle#1$}}

\newcommand{\mcm}[3]{\newcommand{#1}[#2]{{\ensuremath{#3}}}} 
\mcm{\Nbb}{0}{\mathbb{N}}
\mcm{\Zbb}{0}{\mathbb{Z}}
\mcm{\Rbb}{0}{\mathbb{R}}
\mcm{\Cbb}{0}{\mathbb{C}}
\mcm{\Qbb}{0}{\mathbb{Q}}
\mcm{\Acal}{0}{\cal A}
\mcm{\Bcal}{0}{\cal B}
\mcm{\Ccal}{0}{\cal C}
\mcm{\Dcal}{0}{\cal D}
\mcm{\Ecal}{0}{\cal E}
\mcm{\Fcal}{0}{\cal F}
\mcm{\Gcal}{0}{\cal G}
\mcm{\Hcal}{0}{\cal H}
\mcm{\Ical}{0}{\cal I}
\mcm{\Jcal}{0}{\cal J}
\mcm{\Kcal}{0}{\cal K}
\mcm{\Lcal}{0}{\cal L}
\mcm{\Mcal}{0}{\cal M}
\mcm{\Ncal}{0}{\cal N}
\mcm{\Ocal}{0}{{\cal O}}
\mcm{\Pcal}{0}{{\cal P}}
\mcm{\Qcal}{0}{{\cal Q}}
\mcm{\Rcal}{0}{{\cal R}}
\mcm{\Scal}{0}{{\cal S}}
\mcm{\Tcal}{0}{{\cal T}}
\mcm{\Ucal}{0}{{\cal U}}
\mcm{\Vcal}{0}{{\cal V}}
\mcm{\Wcal}{0}{{\cal W}}
\mcm{\Xcal}{0}{{\cal X}}
\mcm{\Ycal}{0}{{\cal Y}}
\mcm{\Zcal}{0}{{\cal Z}}

\usepackage{microtype}
\DeclareMathOperator{\col}{\mathrm{col}}
\DeclareMathOperator{\red}{\textls[-70]{\texttt{red}}}
\DeclareMathOperator{\blue}{\textls[-120]{\texttt{blue}}}

\definecolor{cb_orange}{rgb}{0.859 0.427 0.0}
\definecolor{cb_blue}{rgb}{0.0 0.427 0.859}
\definecolor{cb_violet}{rgb}{0.714 0.427 1.0}
\definecolor{cb_dark_seal}{rgb}{0.0 0.286 0.286}
\definecolor{cb_seal}{rgb}{0.0 0.573 0.573}
\definecolor{cb_pink}{rgb}{1.0 0.427 0.714}
\definecolor{cb_rose}{rgb}{1.0 0.714 0.859}
\definecolor{cb_purple}{rgb}{0.286 0.0 0.573}
\definecolor{cb_light_blue}{rgb}{0.427 0.714 1.0}
\definecolor{cb_vlight_blue}{rgb}{0.714 0.859 1.0}
\definecolor{cb_red}{rgb}{0.573 0.0 0.0}
\definecolor{cb_brown}{rgb}{0.573 0.286 0.0}
\definecolor{cb_green}{rgb}{0.141 1.0 0.141}
\definecolor{cb_yellow}{rgb}{1.0 1.0 0.427}

\begin{document}

\maketitle

\begin{abstract}
 Rank aggregation seeks a representative permutation for a collection of rankings and plays a central role in areas such as social choice, information retrieval, and computational biology.
    Two fundamental aggregation tasks are the \emph{center} and \emph{median} problems, which minimize the maximum and the total distance to the input permutations, respectively.
    While these problems are well understood under Kendall's tau and related distances, their parameterized complexity under the Ulam metric, an edit-distance-based metric on permutations, has remained largely unexplored.
    
    In this work, we initiate a systematic study of the parameterized complexity of rank aggregation under the Ulam metric.
    We consider both the center and median problems, as well as their generalizations to the \emph{$k$-center} and \emph{$k$-median} clustering settings, parameterized by the number of centers $k$ and the distance budget $d$ (corresponding to the maximum distance for center variants and the total distance for median variants). Both problems are known to be NP-hard already for $k=1$.
    
    We show that the Ulam $k$-center problem remains NP-hard when $d=1$, but is fixed-parameter tractable when parameterized by $k + d$.
    Our algorithm is based on a novel local-search framework tailored to the non-local nature of Ulam distances.
    We complement this by proving that no polynomial kernel exists for the $k+d$ parameterization unless $\NP \subseteq \coNP/\poly$.
    For the Ulam \emph{$k$-median} problem parameterized by the total distance $d$, we establish $\W[1]$-hardness and provide an $\XP$ algorithm.
    We also provide a polynomial kernel for the parameter $k + d$, which in turn yields a fixed-parameter tractable algorithm.
\end{abstract}

\newpage
\section{Introduction}\label{sec:intro}
    Rank aggregation is the problem of computing a representative ranking from a collection of permutations representing the preferences of individual voters.
The roots for a mathematical investigation of this problem can be traced back to early studies of social choice theory and voting systems like Borda (1781) and Condorcet (1785), as well as later seminal works such as the ones by Arrow~\cite{arrow1951social} or Diaconis and Graham~\cite{diaconisGraham}.
Rank aggregation constitutes a fundamental task with broad applications in social choice theory, information retrieval, computational biology, and database systems~\cite{aggarwal1988io,arrow1951social,dwork2001rank,pevzner2000computational,DasK25}.
Algorithmically, two canonical rank aggregation tasks are to compute the center and the median of permutations.
Given a distance function $\mathrm{dist}(\cdot,\cdot)$ on the space $\Scal_n$ of permutations over $n$ symbols, the \emph{center} of a permutation set
$\Pi \subseteq \Scal_{n}$ is a permutation $\sigma \in \Scal_n$ minimizing
\begin{equation*}
    \max_{\pi \in \Pi} \mathrm{dist}(\pi, \sigma),
\end{equation*}
whereas the \emph{median} of $\Pi$ is a permutation $\sigma \in \Scal_n$ minimizing
\begin{equation*}
    \sum_{\pi \in \Pi} \mathrm{dist}(\pi, \sigma).
\end{equation*}

Among distance measures for rank aggregation, Kendall's tau and the Ulam metric are the two most natural and fundamental metrics, extensively studied in both theory and applications~\cite{ChakrabortyD0S22fairRanking}.
Kendall's tau distance counts the number of pairwise disagreements between two permutations, thereby capturing local inconsistencies in ranking order.
In contrast, the Ulam metric is an edit distance that measures the minimum number of element move operations (deletion and insertion) required to transform one permutation into another.
This metric captures global structural differences by isolating the longest common subsequence (the elements that are not moved) and quantifying how many elements must be relocated to achieve alignment.

The tasks of computing centers and medians under Kendall's tau metric are well understood from different computational perspectives~\cite{ailon2008aggregating,therese2009crossings,dwork2001rank,fagin2003similarity,kenyon2007rank,pevzner2000computational,vanzuylen2007deterministic}.
Finding the optimal center or median under Kendall's tau is $\NP$-hard already if there are just four input permutations~\cite{therese2009crossings,dwork2001rank}.
Additionally, the median under Kendall's tau, also known as \emph{Kemeny rankings}, admits a polynomial-time approximation scheme (PTAS)~\cite{kenyon2007rank}.
Both versions of rank aggregation under Kendall's tau admit fixed-parameter tractable (FPT) algorithms parameterized by the solution cost~\cite{BachmeierBGH_2015,betzler2009kemeny}, and other natural parameters~\cite{betzler2010dodgson, betzler2009parameterized,christian2007lobbying, CunhaSS24,fernau2014socialchoice}.

In contrast, the study of rank aggregation under the Ulam metric is relatively unexplored, especially from the perspective of parameterized complexity.
A key reason is that Ulam distance is a global measure whose combinatorial structure is inherently harder to exploit algorithmically than the local nature of Kendall's tau.
Consequently, some algorithmic work has focused on estimating the distance itself.
Andoni and Nguyen~\cite{AndoniN10} and Naumovitz \textit{et al.}~\cite{NaumovitzSS17} developed sublinear-time approximation schemes to estimate the Ulam distance between two permutations, using connections to calculate the longest common subsequence.

Finding an optimal center under the Ulam metric was shown to be $\NP$-hard over a decade ago~\cite{BachmeierBGH_2015}, while the median problem was only recently proven $\NP$-hard~\cite{FischerGHK_2025}.
A folklore $2$-approximation exists for both the median and center variants under the Ulam metric.
Recently, Chakraborty \textit{et al.}~\cite{chakraborty2021center} obtained a $(1.5 - \varepsilon)$-approximation (with any $\varepsilon > 0$) for the center problem with running time exponential in $m$ (the number of input permutations).
At a similar time, the first approximation algorithm to break the barrier of $2$ for the median problem appeared~\cite{chakraborty2021ulam}.
The subsequent work~\cite{chakraborty2023clustering} extends the approximation framework to the setting of finding $k$ medians, achieving a similar approximation guarantee in time $(k\log(mn))^{\Ocal(k)}mn^3$, where $n$ is the size of the alphabet and $m$ is the number of input permutations.
Very recently, Jaiswal\textit{et al.}~\cite{jaiswal_et_al:LIPIcs.ICALP.2025.100} improved on this work and provide a randomized algorithm running in $\Tilde{O}((2k)^k nd)$ time.
All three works on the median problem provide guarantees close to but below an approximation factor of $2$.

To the best of our knowledge, no prior work has considered Ulam-based aggregation through parameterized exact algorithms, despite the metric's intimate connection to edit distance and longest common subsequence—domains where parameterized techniques have proven highly effective.
The clustering problems under the Ulam metric seem more challenging, as the Ulam metric is a more complex and less local measure than Kendall's tau, which may suggest that the computational landscape of Ulam-based aggregation could be significantly different.
One may notice that the Ulam distance is a smaller parameter than Kendall's tau, as every permutation can be transformed into the other by only moving the elements that are in a different order in the two permutations.
It may happen that two permutations with Ulam distance $1$ could differ in the relative order of an unbounded number of pairs of elements.
This highlights a fundamental structural disparity between the two distance measures, which leads to the following natural question:
\begin{quote}
     \textit{What is the parameterized complexity of rank aggregation under the Ulam metric?}
\end{quote}
   
\subsection{Clustering Under the Ulam Metric: Definitions and Main Results}\label{sec:results}

In this work, we initiate a systematic study of the parameterized complexity of rank aggregation under the Ulam metric.
Beyond the classical parameter $d$, which captures the solution cost (the maximum distance or the sum of distances, respectively), we consider substantially more general variants of the problem: $k$-center and $k$-median clustering under the Ulam metric~\cite{chakraborty2023clustering}.
These extensions, denoted by \ucenterfull{} and \umedianfull{}, are both natural and well established in the rank aggregation literature.
In these problems, one seeks $k$ representative permutations that minimize, respectively, the maximum distance to any assigned permutation ($k$-center) or the total sum of distances ($k$-median).

Formally, let $\Scal_n(\Sigma)$ be the set of all permutations over an alphabet $\Sigma$ of size $n$, and let $\dist(\pi,\sigma)$ be the Ulam distance between $\pi$ and $\sigma$ in $\Scal_n(\Sigma)$ (see \Cref{sec:preliminaries} for the formal definitions).
We define our problems of interest as follows.

\defproblema{\ucenterfull}
{An alphabet $\Sigma$ of size $n$, a set of permutations $\Pi = \{\pi_1,\pi_2,\dots,\pi_m\} \subseteq \Scal_{n}(\Sigma)$, and integers $k$ and $d$.}
{Decide whether there exists a set $S = \{\sigma_1, \sigma_2, \ldots, \sigma_k \} \subseteq \Scal_n(\Sigma)$ of $k$ centers such that
\begin{equation*}
    \max_{\pi \in \Pi} \min_{\sigma \in S} \dist(\pi,\sigma) \le d.
\end{equation*}}

\defproblema{\umedianfull}
{An alphabet $\Sigma$ of size $n$, a multiset of permutations $\Pi = \{\pi_1,\pi_2,\dots,\pi_m\} \subseteq \Scal_n(\Sigma)$, and
integers $k$ and $d$.}
{Decide whether there exists a set $S = \{\sigma_1, \sigma_2, \ldots, \sigma_k\} \subseteq \Scal_n(\Sigma)$ of $k$ medians such that
\begin{equation*}
    \sum_{\pi \in \Pi} \min_{\sigma \in S} \dist(\pi,\sigma) \le d.
\end{equation*}}

We remark that, in both clustering problems, the centers or medians are not required to belong to the input set of permutations~$\Pi$.
Hence, we consider variants of the problems that are referred to as \emph{continuous} in some literature.
Both problems are known to be $\NP$-hard already for $k = 1$: hardness for the center variant was shown by Bachmaier \textit{et al.}~\cite{BachmeierBGH_2015} and for the median by Fischer \textit{et al.}~\cite{FischerGHK_2025}.

We summarize our main results in \Cref{tab:summary-results}.
Most notably, we establish fixed-parameter tractability for both problems when parameterized simultaneously by $k$ and~$d$.
Prior to our work, the parameterized complexity of these two clustering problems with respect to $d$ was not even settled for the much simpler case $k = 1$.

We contrast this tractability result by establishing hardness for parameterization by $d$ alone: \ucenterfull is already $\NP$-hard on instances with $d=1$.
In contrast, \umedianfull is in $\XP$ parameterized by $d$ (that is, solvable in polynomial time for each fixed $d$), but is $\W[1]$-hard for parameter $d$, ruling out fixed-parameter tractability under well-established complexity assumptions.
Finally, we discover yet another difference in the tractability landscapes of \ucenterfull and \umedianfull.
While the latter admits a polynomial-sized kernel, we rule out the existence of such a kernel for \ucenterfull under established complexity assumptions.
Our results for the Ulam metric complement the mature theory developed for Kendall's tau.
We also hope that they provide rigorous foundations for consensus ranking in applications where global structural coherence under the Ulam metric is essential.

In the following, we provide a detailed overview of our results and the techniques employed.
Formal proofs are deferred to Sections~\ref{sec:center}~and~\ref{sec:median}.

{\crefname{theorem}{Thm.}{Thms.}
 \crefname{corollary}{Cor.}{Cors.}
 
\begin{table}[h]
  \centering
  \caption{Summary of our results.}
  \label{tab:summary-results}
  \begin{tabular}{lcc}
    \toprule
    Parameters  & $k$-center & $k$-median \\
    \midrule
    $k$         & $\NP$-hard even if $k=1$ \cite{BachmeierBGH_2015}
                & $\NP$-hard even if $k=1$ \cite{FischerGHK_2025} \\
    $d$         & $\NP$-hard even if $d=1$ (\cref{thm:center_d_paraNPhard})
                & $\W[1]$-hard, $\XP$ (\cref{thm:median_d_whard,thm:median_d_xp}) \\
    $k+d$       & $\FPT$ (\cref{thm:center_dk_fpt})
                & $\FPT$ (\cref{cor:fpt-median})\\
                & no polynomial kernel (\cref{thm:center_dk_no_poly_kernel})
                & polynomial kernel (\cref{thm:median_dk_fpt})\\
    \bottomrule
  \end{tabular}
\end{table}
}

\subsection{Contributions and Techniques for \texorpdfstring{\ucenterfull}{Ulam Metric k-Center}}

The first group of results in our work concerns \ucenterfull.
As previous work already established that it is unlikely to obtain polynomial time algorithms even for instances aiming for $k = 1$ center, it is natural to ask for the tractability on instances where we instead restrict the radius of the centers to $d = 1$.
We answer this question negatively by establishing the $\NP$-hardness of \ucenterfull for $d$ is any constant at least $1$.
This hardness result is established via a reduction from the \textsc{Vertex Cover} problem on triangle-free graphs. 

\begin{restatable}{theorem}{centerNPhard}
    \label{thm:center_d_paraNPhard}
    \ucenterfull is $\NP$-hard even if $d$ is fixed to any constant at least $1$.
\end{restatable}

The construction for $d = 1$ is relatively straightforward but closes an important gap: it establishes that parameterization by $d$ alone cannot yield fixed-parameter tractability under standard complexity assumptions.
The main idea of the reduction is to construct an alphabet consisting of two symbols $v, \bar{v}$ per vertex $v$ in the vertex cover instance.
The input set $\Pi$ contains one permutation per edge, all sharing a common base ordering except that the symbol pairs of the two endpoints of the edge are swapped in the respective permutation.
A vertex cover then corresponds to a set of center permutations, each encoding a vertex by swapping the respective symbol pair in the base ordering.
Then each permutation encoding an edge has Ulam distance $1$ to only those permutations encoding vertices to which the edge is incident, which ensures an equivalence between the vertex cover instance and the \ucenterfull instance.
We defer a formal proof to Section~\ref{sec:center}, but provide an exemplary reduction in Figure~\ref{fig:exampleVertexCoverReduction}.

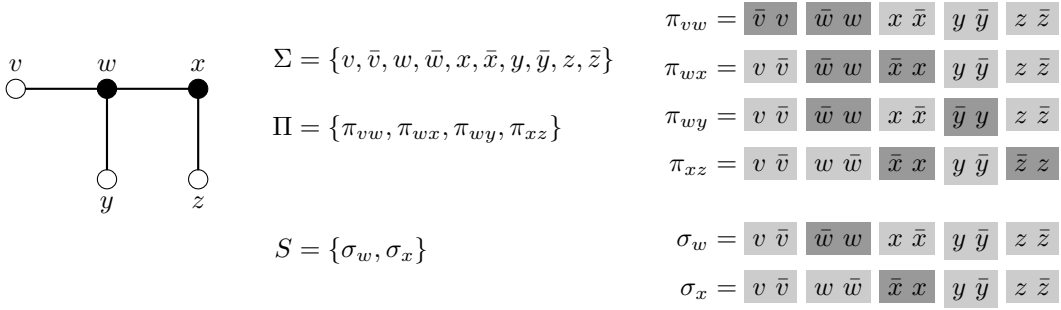
\begin{figure}[t]
    \centering

    \begin{minipage}{.21\textwidth}
        \begin{tikzpicture}[scale = 1.2]
            \draw (0, 2) circle (3pt) node[black,above=3pt] {$v$}; 
            \filldraw (1, 2) circle (3pt) node[black,above=3pt] {$w$}; 
            \filldraw (2, 2) circle (3pt) node[black,above=3pt] {$x$}; 
            \draw (1, 1) circle (3pt) node[black,below=3pt] {$y$}; 
            \draw (2, 1) circle (3pt) node[black,below=3pt] {$z$}; 

            \draw[thick] (1.1, 2) -- (1.9, 2);
            \draw[thick] (1, 1.9) -- (1, 1.1);
            \draw[thick] (2, 1.9) -- (2, 1.1);
            \draw[thick] (0.1, 2) -- (0.9, 2);
        \end{tikzpicture}
    \end{minipage}
    \begin{minipage}{.36\textwidth}
        \begin{align*}
            \Sigma &{}= \{v, \bar{v}, w, \bar{w}, x, \bar{x}, y, \bar{y}, z, \bar{z}\}\\[1em]
            \Pi &{}= \{\pi_{vw}, \pi_{wx}, \pi_{wy}, \pi_{xz}\} \\[3em]
            S &{}= \{\sigma_w, \sigma_x\}
        \end{align*}
    \end{minipage}
    \begin{minipage}{.41\textwidth}
        \begin{align*}
            \pi_{vw} &{}= \HIGHlight{\bar{v}\ v}\ \HIGHlight{\bar{w}\ w}\ \highlight{x\ \bar{x}}\ \highlight{y\ \bar{y}}\ \highlight{z\ \bar{z}}\\
            \pi_{wx} &{}= \highlight{v\ \bar{v}}\ \HIGHlight{\bar{w}\ w}\ \HIGHlight{\bar{x}\ x}\ \highlight{y\ \bar{y}}\ \highlight{z\ \bar{z}}\\
            \pi_{wy} &{}= \highlight{v\ \bar{v}}\ \HIGHlight{\bar{w}\ w}\ \highlight{x\ \bar{x}}\ \HIGHlight{\bar{y}\ y}\ \highlight{z\ \bar{z}}\\
            \pi_{xz} &{}= \highlight{v\ \bar{v}}\ \highlight{w\ \bar{w}}\ \HIGHlight{\bar{x}\ x}\ \highlight{y\ \bar{y}}\ \HIGHlight{\bar{z}\ z}\\[1em]
            \sigma_w &{}= \highlight{v\ \bar{v}}\ \HIGHlight{\bar{w}\ w}\ \highlight{x\ \bar{x}}\ \highlight{y\ \bar{y}}\ \highlight{z\ \bar{z}}\\
            \sigma_x &{}= \highlight{v\ \bar{v}}\ \highlight{w\ \bar{w}}\ \HIGHlight{\bar{x}\ x}\ \highlight{y\ \bar{y}}\ \highlight{z\ \bar{z}}
        \end{align*}
    \end{minipage}
    \caption{Exemplary vertex cover instance with solution $\{w,x\}$ (black vertices) and the constructed instance of \ucenterfull with solution $S = \{\sigma_w, \sigma_x\}$.}
        \label{fig:exampleVertexCoverReduction}
\end{figure}

A nontrivial challenge is that a single permutation—swapping symbol pairs of multiple vertices—might cover several edges without corresponding to any single vertex in a cover.
We circumvent this problem by basing our reduction on triangle-free instances of \textsc{Vertex Cover}, so this case will never occur.
Furthermore, a direct reduction from the $d = 1$ case to general $d$ is non‑trivial; we overcome this by augmenting the $d = 1$ construction with additional corresponding symbols for each vertex.

The $\NP$-hardness of \ucenterfull even if $d = 1$ or if $k = 1$ naturally raises the question of the parameterized complexity of the problem with the combined parameter $k + d$.
The following theorem establishes fixed-parameter tractability with respect to this combined parameter. 

\begin{restatable}{theorem}{centerFPTdk}
    \label{thm:center_dk_fpt}
    \ucenterfull can be solved in $2^{\Ocal(d^{2} k(d + k))} \cdot m n^{1 + o(1)}$  time, where $m\coloneqq |\Pi| $ and $n\coloneqq |\Sigma|$. 
\end{restatable}

On first glance, it may seem that the finite number of permutations in our space would allow for an approach based on some simple branching and \emph{local search} around some permutation $\pi \in \Pi$.
This approach is infeasible, however, as for a permutation $\pi$ of length $n \gg d$, there are $n^{\Omega(d)}$ permutations within Ulam distance $d$ from $\pi$, each of which is a potential candidate for the center.
This lower bound follows from a simple counting argument.
One can choose $d$ symbols from the first half of $\pi$ in $\binom{n/2}{d}$ ways and insert them into the second half in at least $(n/2)^d$ distinct ways, yielding a unique permutation at Ulam distance at most $d$ from $\pi$.

A natural alternative is to adapt techniques from string-based clustering, which typically perform bounded local search around a candidate solution and iteratively improve it by resolving ``dissimilarities'' with inputs that exceed the target distance $d$. 
Such a paradigm is well-known and has been highly successful for problems such as \cstringfull ~\cite{GrammNR03} where the ``dissimilarities'' correspond to bits contributing to the Hamming distance.
However, two permutations can admit $\Omega((n/d)^{d})$ distinct longest common subsequences ($\LCS$) of the same maximum length, making it impossible to pinpoint a canonical set of elements to guide local improvements.
Thus, such techniques from string-based clustering cannot directly apply to the Ulam metric, and more involved methods are required.
Our algorithm embarks on a similar approach as the one by Gramm \textit{et al.}~\cite{GrammNR03}. Rather sooner than later, however, the issue of not having a unique withness $\LCS$ requires much more involved techniques and branching procedures tailored to the Ulam metric to establish fixed-parameter tractability.

In our algorithm, we construct a family of \emph{candidate permutations} that are at bounded Ulam distances from hypothetical centers.
These permutations will eventually converge towards the center permutations.
If $k=1$, then such a candidate can be initially chosen as an arbitrary permutation from $\Pi$; note that it should be at the Ulam distance at most $d$ from a center in any \yesinstance of the problem.
For $k\geq 2$, we initiate the set of candidates using the idea behind the standard $2$-approximation for the \textsc{$k$-Center} problem.
On each step, we either conclude that the current set of candidates can serve as centers, that is, each permutation in $\Pi$ is at distance at most $d$ from one of the candidates, or find a \emph{guide permutation} $\pi_g$ at distance at least $d+1$ from all the candidates.
Then we can guess a candidate permutation $\pi_c$ that should be transformed to a center for the cluster containing the guide.
Our crucial result (\Cref{lem:center_dk_fpt_guide}) shows that, given $\pi_c$ and $\pi_g$, we can efficiently (in FPT time) enumerate a set of permutations with distance one to $\pi_c$ such that for any hypothetical center, there is a permutation in the set that is closer to the center than $\pi_c$ was, see Figure~\ref{fig:kcenter-sigma} for a visualization of this idea.
The proof of \Cref{lem:center_dk_fpt_guide} is rather sophisticated and combines the \emph{random separation}~\cite{CaiCC06} and branching techniques.
The random separation is used to highlight elements of $\pi_c$ and $\pi_g$ that should be moved to make these permutations closer to the center and the neighbors of these elements that are not moved.
Then we identify an element of $\pi_c$ that should be moved and, after that, find an appropriate place for this element using branching.
Repeatedly invoking \Cref{lem:center_dk_fpt_guide}, we ensure that the candidate permutations converge to the centers.

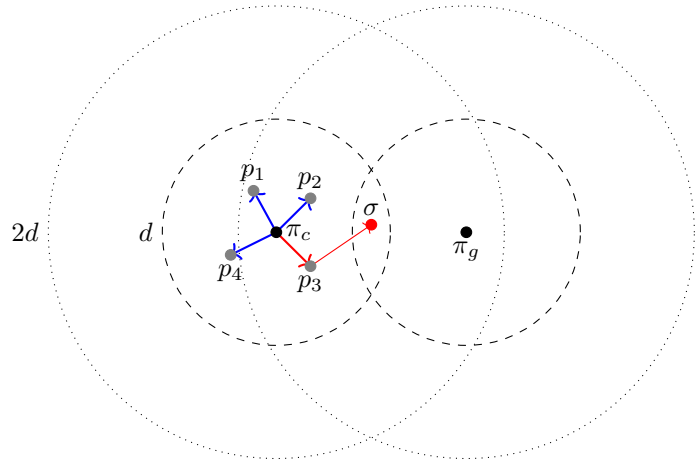
\begin{figure}[t]
    \centering
    \begin{tikzpicture}[scale=1]

        \coordinate (pc) at (0,0);
        \coordinate (pg) at (2.5,0); 
        \coordinate (sigma) at (1.25,0.1); 

        \coordinate (p1) at (-0.3,0.55);
        \coordinate (p2) at (0.45,0.45);
        \coordinate (p4) at (-0.6,-0.3);
        \coordinate (pstar) at (0.45,-0.45); 

        \node[right] at (pc) {$\pi_c$};
        \node[below] at (pg) {$\pi_g$};
        \node[above] at (sigma) {$\sigma$};

        \node[above]  at (p1) {$p_1$};
        \node[above] at (p2) {$p_2$};
        \node[below] at (p4) {$p_4$};
        \node[below] at (pstar) {$p_3$};

        \draw[dashed] (pc) circle (1.5);     
        \draw[dotted] (pc) circle (3);       

        \draw[dashed] (pg) circle (1.5);     
        \draw[dotted] (pg) circle (3);       

        \node[left] at (-1.5,0) {$d$};
        \node[left] at (-3.0,0) {$2d$};

        \draw[->, thick, blue] (pc) -- (p1);
        \draw[->, thick, blue] (pc) -- (p2);
        \draw[->, thick, blue] (pc) -- (p4);
        
        \draw[->, thick, red] (pc) -- (pstar);
        \draw[->, red] (pstar) -- (sigma);

        \filldraw[black] (pc) circle (2pt);
        \filldraw[black] (pg) circle (2pt);
        \filldraw[red]   (sigma) circle (2pt);

        \filldraw[gray] (p1) circle (2pt);
        \filldraw[gray] (p2) circle (2pt);
        \filldraw[gray] (p4) circle (2pt);
        \filldraw[gray] (pstar) circle (2pt);

    \end{tikzpicture}
    \caption{
        Illustration of the key idea behind our FPT algorithm for the $k$-center problem under Ulam distance.
        $\pi_c$ denotes a candidate center, $\pi_g$ a guide permutation with distance between $d$ and $2d$ from $\pi_c$, and $\sigma$ represents the true (optimal) center that covers both $\pi_c$ and $\pi_g$.
        We enumerate in FPT time a set of permutations (shown as $\{p_1, p_2, p_3, p_4\}$) at distance $1$ from $\pi_c$.
        Among these candidates, at least one candidate (shown as $p_3$) reduces the Ulam distance to $\sigma$ and thus corresponds to a correct branch.
        }
    \label{fig:kcenter-sigma}
\end{figure}

\medskip
The fixed-parameter tractability of \ucenterfull with parameter  $k+d$ established in \Cref{thm:center_dk_fpt} immediately leads to the question about the existence of a polynomial kernel for this problem. The following theorem provides a negative answer to this question. 

\begin{restatable}{theorem}{centernokernel} 
    \label{thm:center_dk_no_poly_kernel}
    \ucenterfull does not admit a polynomial-sized kernel for the parameter $k+d$ unless $\mathrm{NP} \subseteq \mathrm{coNP}/ \mathrm{poly}$.
\end{restatable}
The lower bound of \Cref{thm:center_dk_no_poly_kernel} holds already for $k=1$.
The proof is based on a Polynomial Parameterized Transformation (PPT) reduction from the classical $\cstringfull$ problem to \ucenterfull (\Cref{lem:red_string_to_center}).
To build such a transformation, we employ a \emph{distance-preserving} transformation between binary strings and permutations, where for every pair of binary strings their Hamming distance is the same as the Ulam distance between the permutations to which they translate.
The reduction establishes that from the perspective of parameterized complexity with respect to the radius $d$, $\ucenterfull$ is at least as hard as $\cstringfull$.
\Cref{thm:center_dk_no_poly_kernel} then follows by pipelining the PPT with established lower bounds on the kernelization of $\cstringfull$ by Basavaraju \textit{et al.}~\cite{BASAVARAJU201821}.

\subsection{Contributions and Techniques for \texorpdfstring{\umedianfull}{Ulam Metric k-Median}}

 The second group of results is concerned with \umedianfull. 
 Recall that the problem is $\NP$-hard for $k=1$ \cite{BachmeierBGH_2015}.
 While \ucenterfull is also $\NP$-hard for $d=1$, this is not the case for \umedianfull.
 It is easy to obtain a polynomial-time algorithm for this problem for each fixed value of $d$. 
 
  \begin{restatable}{theorem}{medianXP}
 \label{thm:median_d_xp}
     \umedianfull can be solved in $\Ocal((mn)^{2d})$ time, where $m\coloneqq |\Pi| $ and $n\coloneqq |\Sigma|$.
\end{restatable}

The algorithm builds on a simple branching technique, where we observe that the medians can be obtained by applying a total of $d$ move-operations on the input set, and we branch on these $d$ moves by brute-force.
The result in \Cref{thm:median_d_xp} is tight in the sense that the existence of an \FPT-algorithm for the parameterization by $d$ is unlikely due to the following lower bound.
 
\begin{restatable}{theorem}{medianWhard}\label{thm:median_d_whard}
    \umedianfull is $\W[1]$-hard parameterized by $d$.
\end{restatable}

Our proof of \cref{thm:median_d_whard} relies on a careful reduction from \textsc{Multicolored Clique}.
The instance is designed in a way that all but one of the medians in any hypothetical solution are already present in the input set $\Pi$.
The challenge is to find that additional median in a way that it is sufficiently close to a number of input permutations. 
Crucially, the permutations are designed in such a way that this is only possible if all edges covered by the median form a clique, see Figure~\ref{fig:multicolored-clique} for a schematic example.
On a high level, we construct the instance such that there is one special symbol per color, and the position of that symbol in the additional median encodes which vertex of that color is selected for the multi-colored clique. 

\newcommand{\fourpoints}[1]{%
    \node[draw=none] at (-2em,0) {#1};
  \foreach \i in {0,...,3} {%
    \fill (\i*1em,0) circle (3pt);
  }%
}

\newcommand{\fourpointsuline}[1]{%
    \fourpoints{#1}
  \draw[line width=1pt, color=gray] 
    (-.5em,-0.5em) -- (3.5em,-0.5em);
}

\newcommand{\vw}{{\color{cb_blue}w}}
\newcommand{\vx}{{\color{cb_red}x}}
\newcommand{\vy}{{\color{cb_orange}y}}
\newcommand{\vz}{{\color{cb_violet}z}}
\newcommand{\vyy}{{\color{cb_orange}y'}}
\newcommand{\vzz}{{\color{cb_violet}z'}}

\begin{figure}[ht]
    
    \begin{minipage}{.34\textwidth}
    \begin{tikzpicture}[
        every node/.style={circle, fill, inner sep = 2pt, minimum size = 7pt},
        scale=.8]

            \def\dx{2} 
            
            \node (w)[color=cb_blue, label={above:$\vw$}] at (0,\dx) {};
            \node (x)[color=cb_red, label={above:$\vx$}] at (\dx,\dx) {};
            \node (y)[color=cb_orange, label={below:$\vy$}] at (0,0) {};
            \node (z)[color=cb_violet, label={below:$\vz$}] at (\dx,0) {};
            \node (y1)[color=cb_orange, label={above:$\vyy$}] at (-.6*\dx,\dx) {};
            \node (z1)[color=cb_violet, label={above:$\vzz$}] at (1.6*\dx,\dx) {};
            
            \draw (w) -- (x);
            \draw (w) -- (y1);
            \draw (w) -- (z);
            \draw (w) -- (y);
            
            \draw (x) -- (z1);
            \draw (x) -- (z);
            \draw (x) -- (y);
            
            \draw (y) -- (z);

            \draw[rounded corners = 5ex, dotted] (-.6, 3) rectangle (2.6, -1);
            \node[fill = none] at (1, 2.6) {$Q$};

        \end{tikzpicture}
    \end{minipage}
    \begin{minipage}{.65\textwidth}
        \begin{tikzpicture}[
            every node/.style={circle, draw, minimum size=18pt},
            scale=.8]

            \begin{scope}[shift={(15,2)}]
                \draw[rounded corners] (-7.5em, 1em) rectangle (4.5em, -1em);
                \fourpointsuline{$\pi_{\vw\vyy} = \sigma_{\vw\vyy}\quad\quad\quad$}
            \end{scope}
            
            \begin{scope}[shift={(15,0)}]
                \draw[rounded corners] (-7.5em, 1em) rectangle (4.5em, -1em);
                \fourpointsuline{$\pi_{\vy\vzz} = \sigma_{\vy\vzz}\quad\quad\quad$}
            \end{scope}

            \draw[rounded corners] (5.6,3) rectangle (11.6,-1);

            \begin{scope}[shift={(7,1.5)}]
                \fourpoints{$\pi_{\vw\vy}$}
            \end{scope}
            \begin{scope}[shift={(10,1.5)}]
                \fourpoints{$\pi_{\vw\vy}$}
            \end{scope}
            \begin{scope}[shift={(7,0.5)}]
                \fourpoints{$\pi_{\vw\vz}$}
            \end{scope}
            \begin{scope}[shift={(10,0.5)}]
                \fourpoints{$\pi_{\vy\vy}$}
            \end{scope}
            \begin{scope}[shift={(7,-0.5)}]
                \fourpoints{$\pi_{\vy\vz}$}
            \end{scope}
            \begin{scope}[shift={(10,-0.5)}]
                \fourpoints{$\pi_{\vy\vz}$}
            \end{scope}

            \begin{scope}[shift={(8.5,2.5)}, color=gray]
                \fourpointsuline{\large$\sigma_{Q}$}
            \end{scope}
        
        \end{tikzpicture}
    \end{minipage}
    \caption{Illustration of the reduction from \textsc{Multicolored Clique} to \umedianfull.
        Left shows an instance of \textsc{Multicolored Clique} with four color classes ($\{w\}, \{x\}, \{y,y'\},$ and $\{z,z'\}$) and a colorful clique $Q = \{w, x, y, z\}$.
        Right shows the schematic solution to the constructed \umedianfull instance.
        Each edge is encoded as an input permutation; edges whose endpoints form a multicolored clique are covered by a single external median $\sigma_{Q}$ (underlined), while all other edges each serve as their own median in singleton clusters (also underlined).}
        \label{fig:multicolored-clique}
\end{figure}
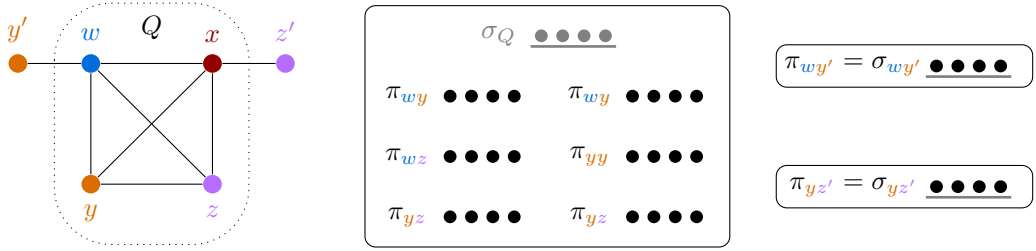

Interestingly, our proof indicates that the hardness of the problem does not so much depend on identifying which symbols need to be moved but also importantly on deciding where to move these symbols.  
 
 In contrast to the lower bound on the polynomial kernel for \ucenterfull obtained in \Cref{thm:center_dk_no_poly_kernel}, \umedianfull admits a polynomial kernel. 
 
\begin{restatable}{theorem}{mediankernel}\label{thm:median_dk_fpt}
    \umedianfull admits  a kernel with $\Ocal(d^2+dk)$ permutations over an alphabet of size $\Ocal(d^4 + d^2k^2)$, computable in $\Ocal(m^2 n \log n)$ time, where $m\coloneqq |\Pi| $ and $n\coloneqq |\Sigma|$. 
\end{restatable}

Our kernelization result is based on the observation that for any \yesinstance 
of \umedianfull, $\Pi$ contains at most $k+d$ distinct permutations. 
 This allows us to reduce the size of $\Pi$. Then we reduce the size of the alphabet $\Sigma$ by making replacements of identical subpermutations by shorter ones. Here, we use the fact that two permutations at Ulam distance at least $d+1$ cannot belong to the same median in any solution. Thus, we can identify some sets of permutations such that every set of permutations that shares a median in a hypothetical solution is a subset to one of the sets. We show that the pairwise Ulam distance within each set are small. Combined with the fact that there are only few distinct permutations, this implies that in each set there are many (or long) identical subpermutations, which we can then shorten to reduce the size of the alphabet $\Sigma$.

Combining \Cref{thm:median_dk_fpt,thm:median_d_xp}, we immediately obtain an \FPT-algorithm with the following running time.

\begin{corollary}\label{cor:fpt-median}
\umedianfull can be solved in $(kd)^{\Ocal(d)} + \Ocal(m^2 n \log n)$ time, where $m\coloneqq |\Pi| $ and $n\coloneqq |\Sigma|$.
\end{corollary}

\section{Preliminaries}\label{sec:preliminaries}
    \subparagraph*{Basic notation.}
We use $[n]$ to denote the set $\{1, 2, \ldots, n\}$.
Let $\Sigma$ be a finite alphabet with $|\Sigma| = n$.  
A \emph{permutation} of length $n$ over $\Sigma$ is a bijection $\pi \coloneq [n] \to \Sigma$.
We write $\Scal_n(\Sigma)$ for the set of all permutations of length $n$ over $\Sigma$.
Thus, each $\pi \in \Scal_n(\Sigma)$ is an ordering of the symbols of~$\Sigma$.
We often write a permutation as a sequence $\pi = \pi(1)\pi(2)\cdots\pi(n)$.

A \emph{subsequence} of a permutation $\pi$ is obtained by deleting zero or more elements from $\pi$ while maintaining the relative order of the remaining elements.
A \emph{common subsequence} of two permutations $\pi_1$ and $\pi_2$ is a subsequence that appears in both permutations.
We denote by $\LCS(\pi_1,\pi_2)$ the length of the longest common subsequence of $\pi_1$ and $\pi_2$.

\subparagraph*{Ulam distance.}
A \emph{move operation} on a permutation $\pi$ removes one element from its current position and inserts it at another position, while maintaining the relative order of all other elements.

For two permutations $\pi,\sigma \in \Scal_n(\Sigma)$, the \emph{Ulam distance} $\dist(\pi,\sigma)$ is the minimum number of move operations needed to transform $\pi$ into $\sigma$:
\begin{equation*}
    \dist(\pi,\sigma) = \min \{\, d : \text{$\pi$ can be turned into $\sigma$ using $d$ moves} \,\}.
\end{equation*}

\noindent A classical identity characterizes this metric in terms of the longest common subsequence:
\begin{equation*}
    \dist(\pi,\sigma) = n - \LCS(\pi,\sigma).
\end{equation*}

Thus, the Ulam distance measures the minimum number of elements that must be
relocated to make the permutations consistent in relative order. We remark that a longest common subsequence, and thereby the Ulam distance, between two permutations of length $n$ can be computed in time in $O(n \log n)$ by relabeling the alphabet such that the first permutation has increasing order and then finding a longest increasing sequence in the other (relabeled) permutation.
The Ulam distance on permutations is analogous to the Levenshtein edit distance
on strings in the following sense: both distances identify a longest common
subsequence and keep it untouched while performing edit operations on the
remaining symbols.

For a set $\Scal\subseteq \Scal_n(\Sigma)$ of permutations and $\pi\in \Scal_n(\Sigma)$, we write $\dist(\pi, \Scal)$ to denote the minimum distance from $\pi$ to any permutation in $\Scal$, that is, $\dist(\pi, \Scal) = \min_{\sigma \in \Scal} \dist(\pi, \sigma)$.

For two strings $s, t$ of equal length over an alphabet, the \emph{Hamming distance} $\distH(s,t)$ counts the number of positions where the corresponding symbols differ.

\subparagraph*{Graphs.}
We work with simple undirected graphs. For a graph $G$, we denote by $V(G)$ its vertex set and by $E(G)$ its edge set.
A \emph{vertex cover} of a graph $G$ is a subset $X \subseteq V(G)$ such that every edge in $E(G)$ has at least one endpoint in $X$.
A \emph{clique} in a graph $G$ is a subset of vertices such that every pair of vertices in the subset is connected by an edge.
A graph is \emph{triangle-free} if it contains no clique of size $3$.
A \emph{proper vertex coloring} of a graph $G$ assigns colors to vertices such that no two adjacent vertices share the same color.  

\subparagraph*{Parameterized Complexity.} We refer to the textbooks~\cite{Cygan2015parameterized,fomin2019kernelization} for an introduction to the parameterized complexity theory and kernelization. 

A \emph{parameterized problem} is a language $L\subseteq\Sigma^*\times\mathbb{N}$  where $\Sigma^*$ is a set of strings over a finite alphabet $\Sigma$.
An input of a parameterized problem is a pair $(x,k)$ where $x$ is a string over $\Sigma$ and $k\in \mathbb{N}$ is a \emph{parameter}. 
A parameterized problem is \emph{fixed-parameter tractable} (or \FPT) if it can be solved in  $f(k)\cdot |x|^{\mathcal{O}(1)}$ time for some computable function~$f$.  
The complexity class \FPT{} consists of all fixed-parameter tractable parameterized problems.
A parameterized problem is in the class \XP{} if it can be solved in $|x|^{f(k)}$ time for a computable function $f(\cdot)$.
The standard way to rule out the existence of an \FPT-algorithm for a parameterized problem under the standard parameterized complexity assumption that $\FPT\neq \W[1]$
is to show that it is  \W[1]-hard, that is, at least as hard as any problem in the class of parameterized problems $\W[1]$.
A problem is para-$\NP$-hard if it $\NP$-hard even when the parameter $k$ is a constant.   

A \emph{kernelization algorithm} or \emph{kernel} for a parameterized problem~$L$ is a polynomial-time algorithm that takes as its input an instance $(x,k)$ of $L$ and returns an instance~$(x',k')$ of the same problem such that (i) $(x,k)\in L$ if and only if $(x',k')\in L$ and (ii)~$|x'|+k'\leq f(k)$ for some computable function~$f\colon \mathbb{N}\rightarrow \mathbb{N}$.
The function $f(\cdot)$ is the \emph{size} of the kernel; a kernel is \emph{polynomial} if $f$ is a polynomial.
A decidable parameterized problem is \FPT{} if and only if it admits a kernel.
However, up to some standard complexity-theoretic assumptions, there are \FPT problems that have no polynomial kernels.
One of the ways to exclude a polynomial kernel for a parameterized problem is to show that there is 
a \emph{polynomial parameter transformation} from a parameterized problem for which the existence of a polynomial kernel does not exist unless $\mathrm{NP} \subseteq \mathrm{coNP} / \mathrm{poly}$, where a polynomial parameter transformation is a polynomial-time reduction such that the parameter in the new instance is bounded by a polynomial in the original parameter.

\section{The Complexity of \texorpdfstring{\ucenterfull}{Ulam Metric k-Center}}\label{sec:center}
    
In this section, we study the \ucenterfull problem.
We first show that the problem is already $\NP$-hard for $d = 1$.
Our main result in this section is the fixed-parameter tractability of the problem parameterized by $k + d$.
Nevertheless, under standard complexity assumptions, the problem does not admit a polynomial kernel parameterized by $k + d$.

\subsection{Lower Bound for Constant Radius \texorpdfstring{$d$}{d}}

We prove that \ucenterfull is $\NP$-hard when the radius is fixed to the constant $d$, via a reduction from \vcfull.
Notice that \vcfull remains $\NP$-hard in triangle-free graphs~\cite{Poljak1974}.
This follows from a standard $2$-subdivision transformation: replacing each edge by a path of length $3$ eliminates all triangles and increases the size of a minimum vertex cover by exactly one per original edge.
We may thus assume without loss of generality that the input graph is triangle-free, which suffices to prove \Cref{thm:center_d_paraNPhard}.
The idea of the reduction is to have two symbols per vertex. All permutations relevant to the reduction look almost entirely the same, except that only the pair of symbols for one or two vertices are swapped to encode that respective vertex or edge.
If we let the input set of permutations encode all edges of the graph this way, then a vertex cover corresponds to a set of permutations encoding vertices.
This holds as each permutation encoding an edge has Ulam distance $1$ to only those permutations encoding vertices to which the edge is incident. 

\centerNPhard*
\begin{proof}
    We start by giving a reduction from \vcfull in triangle-free graphs for the case $d=1$.
    Let $(G, k')$ be an instance of \vcfull such that $G$ is triangle-free and $n'\coloneqq |V(G)|$
    We construct an instance $(\Pi, k, d)$ of $\ucenterfull$  as follows.
    Consider the alphabet set $\Sigma = \{v_i,\bar{v}_{i} : i \in [n'] \}$.
    We first construct a base permutation. Every permutation constructed later will only minimally differ from this permutation.
    \begin{equation*}
        \pi_{\base} \coloneq \highlight{v_1\ \bar{v}_{1}}\ \highlight{v_2\ \bar{v}_{2}} \cdots \highlight{v_{n'}\ \bar{v}_{n'}}.
    \end{equation*}
    Then, for each vertex $v_i \in V$, we introduce a permutation $\sigma_{v_i}$ as the permutation obtained by flipping the positions of $v_i,\bar{v}_{i}$ in $\pi_{\base}$, i.e., 
    \begin{equation*}
        \sigma_{v_i} = \highlight{v_1\ \bar{v}_{1}}\ \highlight{v_2\ \bar{v}_{2}} \cdots \HIGHlight{\bar{v}_{i}\ v_i} \cdots \highlight{v_{n'}\ \bar{v}_{n'}}.
    \end{equation*}
    Similarly, for an edge $e = v_{i}v_{j} \in E(G)$, we also introduce the permutation $\pi_{e}$ obtained from $\pi_{\base}$ by flipping the positions of $v_i, \bar{v}_{i}$ and $v_j, \bar{v}_{j}$, respectively.
    That is, 
    \begin{align*}
        \pi_{e} = \highlight{v_1\ \bar{v}_{1}}\ \highlight{v_2\ \bar{v}_{2}} & \cdots 
        \highlight{v_{i-1}\ \bar{v}_{i-1}}\ 
        \HIGHlight{\bar{v}_{i}\ v_i}\ 
        \highlight{v_{i+1}\ \bar{v}_{i+1}} \cdots \\
        & \cdots \highlight{v_{j-1}\ \bar{v}_{j-1}}\ 
        \HIGHlight{\bar{v}_{j}\ v_j}\ 
        \highlight{v_{j+1}\ \bar{v}_{j+1}} \cdots 
        \highlight{v_{n'}\ \bar{v}_{n'}}.
    \end{align*}
    Next, we set $d=1$, $k=k'$, and let the set of permutations consist of all permutations for edges in the graph, i.e., $\Pi =\{ \pi_{e} :  e\in E(G) \}$.
    Clearly, the instance can be constructed in polynomial time. An exemplary construction is depicted in Figure~\ref{fig:exampleVertexCoverReduction}.

    If $G$ has a vertex cover $X$ of size $k$, then consider $S = \{\sigma_{v_{i}} : v_{i} \in X \}$ as a of center permutations. 
    For every edge $e = v_{i}v_{j} \in E(G)$, $X$ contains at least one of $v_i$ and $v_j$, so that $\sigma_{v_i} \in S$ or $\sigma_{v_j} \in S$.
    In addition, the \ulamdistance between $\pi_{e}$ and $\sigma_{v_i}$ or $\sigma_{v_j}$ is exactly $1$.
    As a result, we have a solution $S$ for our \ucenterfull instance.

    For the reverse direction, suppose we have a \yesinstance for \ucenterfull witnessed by a set $S$ of permutations with $|S| \leq k$, then we create a vertex cover $X$ for $G$ as follows.
    Note that $S$ partitions $\Pi$ into at most $k$ clusters, where we have one cluster per center permutation in $S$, and each permutation in $\Pi$ is assigned to the cluster of the closest center in Ulam distance (breaking ties arbitrarily).
    Consider a cluster $C$.
    We first observe that for any two permutations $\pi_{e_1}, \pi_{e_2} \in C$, the corresponding edges $e_1$ and $e_2$ must share a common endpoint.
    Otherwise, the distance between $\pi_{e_1}$ and $\pi_{e_2}$ would be $4$, and hence there would be no center permutation at distance at most $1$ from both.

    If $|C| \geq 4$, then the only way every pair of edges corresponding to permutations in $C$ can share a common endpoint is that all these edges share a single vertex.
    Otherwise, there would exist a pair of edges with no common endpoint, again yielding distance $4$ between the corresponding permutations.
    In this case, we add the vertex shared by all edges in $C$ to the vertex cover $X$.

    If $|C| = 3$, then the shared-endpoint condition implies that the edges $e_1, e_2,$ and $e_3$ either form a triangle or all share a common endpoint.
    Since $G$ is triangle-free, only the latter is possible and thus, we add the vertex shared by all three edges to the vertex cover $X$.

    Finally, for $|C|=2$, we add the unique vertex shared by the two edges associated with $C$ to the vertex cover $X$ and if $|C|=1$ we add an arbitrary endpoint of the unique edge associated with $C$ to the vertex cover $X$.
    
    Since each permutation in $\Pi$ belongs to some cluster, the edge corresponding to that permutation is covered by the vertex chosen for that cluster.
    The set $X$ of vertices selected by the clusters forms a vertex cover of $G$.
    Moreover, because each cluster contributes exactly one vertex, we have $|X| = |S| \leq k = k'$.
    Thus, our \vcfull instance $(G, k')$ should be a \yesinstance, completing the proof.

    Now we consider any constant $d$ larger than $1$, employ the same construction but extend the alphabet set to $\Sigma = \{v_{i, j}, \bar{v}_{i, j} : i \in [n'], j \in [d] \}$.
    Then, the base permutation $\pi_{\base}$ is defined as
    \begin{equation*}
        \pi_{\base} = \highlight{v_{1,1}\ \bar{v}_{1,1}} \cdots \highlight{v_{1,d}\ \bar{v}_{1,d}}\ \highlight{v_{2,1}\ \bar{v}_{2,1}} \cdots \highlight{v_{2,d}\ \bar{v}_{2,d}} \cdots \highlight{v_{n',1}\ \bar{v}_{n',1}} \cdots \highlight{v_{n',d}\ \bar{v}_{n',d}}.
    \end{equation*}
    Compared to $\pi_{\base}$, let every input permutation $\pi_e$ swap the two sets of $d$ symbols, each corresponding to its incident vertices.
    The Ulam distance to $\pi_{\base}$ then is $2d$.
    Analogously, a solution permutation $\sigma_{v_i}$ now swaps the $d$ pairs of symbols corresponding to vertex $v_i$, compared to $\pi_{\base}$.
    Notice that the Ulam distance between any two permutations $\pi_{e_1}$ and $\pi_{e_2}$ is $4d$ if the corresponding edges $e_1$ and $e_2$ do not share a common endpoint, and the Ulam distance between $\pi_e$ and $\sigma_{v_i}$ is exactly $d$ if and only if $v_i$ is an endpoint of $e$.
    The above arguments still apply, except that any solution now requires a radius of $d$ around the centers.
\end{proof}

\subsection{Fixed-Parameter Tractability by \texorpdfstring{$k + d$}{k+d}}

Since the problem is $\NP$-hard even for constant $d$ or $k$, neither parameter yields tractability on its own.
We therefore consider the combined parameter $k + d$, and show that the Ulam center problem is fixed-parameter tractable under this parameterization.

Our algorithm relies on iteratively updating candidate permutations to eventually converge towards the center permutations by moving one symbol in a candidate in each step.
Specifically, given a \emph{candidate} permutation $\pi_c$ and a \emph{guide} permutation $\pi_g$, we can enumerate (in $\FPT$ time) a set $P$ of permutations that differ from $\pi_{c}$ by one move.
Among them, at least one gets closer to any true center that covers both $\pi_c$ and $\pi_g$.
The key technical ingredient is captured in Lemma~\ref{lem:center_dk_fpt_guide}.

Our approach is based on the $2$-colorings of the symbols in $\pi_g$ and $\pi_c$.
The $2$-colorings encode which symbols are allowed to be moved toward a specific permutation $\sigma \in P$.  
We formalize when such a $2$-coloring is fitting with respect to $\sigma$, an exemplary coloring is visualized in Figure~\ref{fig:dfitting}.

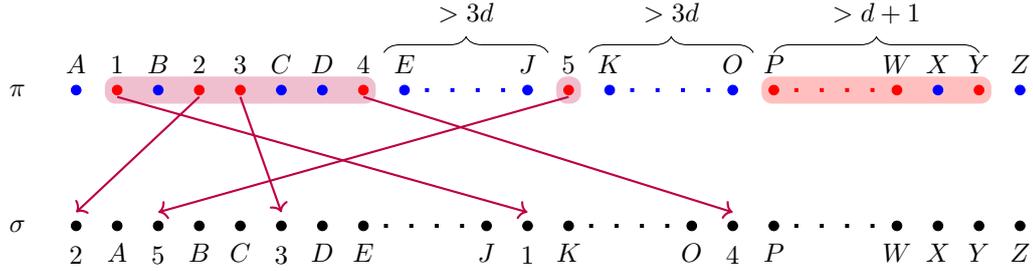
\begin{figure}[h]
\centering
\begin{tikzpicture}[scale=.9]

\def\dx{0.6}

\begin{pgfonlayer}{background}

    \fill[purple!25, rounded corners]
        (1.7*\dx, 1.8) rectangle (8.3*\dx, 2.2);

    \fill[purple!25, rounded corners]
        (12.7*\dx, 1.8) rectangle (13.3*\dx, 2.2);
    
    \fill[red!25, rounded corners]
        (17.7*\dx, 1.8) rectangle (23.3*\dx, 2.2);
    
\end{pgfonlayer}


\filldraw[blue] (1*\dx, 2) circle (2pt) node[black,above=3pt] {$A$}; 
\filldraw[red] (2*\dx, 2) circle (2pt) node[black,above=3pt] {$1$}; 
\filldraw[blue] (3*\dx, 2) circle (2pt)node[black,above=3pt] {$B$}; 
\filldraw[red] (4*\dx, 2) circle (2pt) node[black,above=3pt] {$2$};
\filldraw[red] (5*\dx, 2) circle (2pt) node[black,above=3pt] {$3$};
\filldraw[blue] (6*\dx, 2) circle (2pt) node[black,above=3pt] {$C$}; 
\filldraw[blue] (7*\dx, 2) circle (2pt) node[black,above=3pt] {$D$}; 
\filldraw[red] (8*\dx, 2) circle (2pt) node[black,above=3pt] {$4$};
\filldraw[blue] (9*\dx, 2) circle (2pt) node[black,above=3pt] {$E$};  
\draw[
    color = blue,
    mark=*,
    line width=1.5pt,
    dash pattern=on 1.5pt off 8pt
]
    (9.5*\dx, 2) -- (11.5*\dx, 2) ;

\filldraw[blue] (12*\dx, 2) circle (2pt) node[black,above=3pt] {$J$};; 
\filldraw[red] (13*\dx, 2) circle (2pt) node[black,above=3pt] {$5$};; 
\filldraw[blue] (14*\dx, 2) circle (2pt) node[black,above=3pt] {$K$};; 

\draw[
    color = blue,
    mark=*,
    line width=1.5pt,
    dash pattern=on 1.5pt off 8pt
]
    (14.5*\dx, 2) -- (16.5*\dx, 2);

\filldraw[blue] (17*\dx, 2) circle (2pt) node[black,above=3pt] {$O$};
\filldraw[red] (18*\dx, 2) circle (2pt) node[black,above=3pt] {$P$};

\draw[
    color = red,
    mark=*,
    line width=1.5pt,
    dash pattern=on 1.5pt off 8pt
]
    (18.5*\dx, 2) -- (20.5*\dx, 2);

\filldraw[red] (21*\dx, 2) circle (2pt) node[black,above=3pt] {$W$};
\filldraw[blue] (22*\dx, 2) circle (2pt) node[black,above=3pt] {$X$};
\filldraw[red] (23*\dx, 2) circle (2pt)node[black,above=3pt] {$Y$};
\filldraw[blue] (24*\dx, 2) circle (2pt)node[black,above=3pt] {$Z$};

\node[left] at (0,2) {$\pi$};


\foreach \i in {1,...,24} {
    
}

\filldraw ( 1*\dx, 0 ) circle (2pt); \node[below=7pt] at (1*\dx, 0.1){$2$};
\filldraw ( 2*\dx, 0 ) circle (2pt);\node[below=7pt] at (2*\dx, 0.15){$A$};
\filldraw ( 3*\dx, 0 ) circle (2pt);\node[below=7pt] at (3*\dx, 0.1){$5$};
\filldraw ( 4*\dx, 0 ) circle (2pt);\node[below=7pt] at (4*\dx, 0.15){$B$};
\filldraw ( 5*\dx, 0 ) circle (2pt);\node[below=7pt] at (5*\dx, 0.15){$C$};
\filldraw ( 6*\dx, 0 ) circle (2pt);\node[below=7pt] at (6*\dx, 0.1){$3$};
\filldraw ( 7*\dx, 0 ) circle (2pt);\node[below=7pt] at (7*\dx, 0.15){$D$};
\filldraw ( 8*\dx, 0 ) circle (2pt);\node[below=7pt] at (8*\dx, 0.15){$E$};
\draw[
    mark=*,
    line width=1.5pt,
    dash pattern=on 1.5pt off 8pt
]
    (8.5*\dx, 0) -- (10.5*\dx, 0);

\filldraw ( 11*\dx, 0 ) circle (2pt);\node[below=7pt] at (11*\dx, 0.15){$J$};
\filldraw ( 12*\dx, 0 ) circle (2pt);\node[below=7pt] at (12*\dx, 0.1){$1$};
\filldraw ( 13*\dx, 0 ) circle (2pt);\node[below=7pt] at (13*\dx, 0.15){$K$};
\draw[
    mark=*,
    line width=1.5pt,
    dash pattern=on 1.5pt off 8pt
]
    (13.5*\dx, 0) -- (15.5*\dx, 0);

\filldraw ( 16*\dx, 0 ) circle (2pt);\node[below=7pt] at (16*\dx, 0.15){$O$};
\filldraw ( 17*\dx, 0 ) circle (2pt);\node[below=7pt] at (17*\dx, 0.1){$4$};
\filldraw ( 18*\dx, 0 ) circle (2pt);\node[below=7pt] at (18*\dx, 0.15){$P$};

\draw[
    mark=*,
    line width=1.5pt,
    dash pattern=on 1.5pt off 8pt
]
    (18.5*\dx, 0) -- (20.5*\dx, 0);

\filldraw ( 21*\dx, 0 ) circle (2pt);\node[below=7pt] at (21*\dx, 0.15){$W$};
\filldraw ( 22*\dx, 0 ) circle (2pt);\node[below=7pt] at (22*\dx, 0.15){$X$};
\filldraw ( 23*\dx, 0 ) circle (2pt);\node[below=7pt] at (23*\dx, 0.15){$Y$};
\filldraw ( 24*\dx, 0 ) circle (2pt);\node[below=7pt] at (24*\dx, 0.15){$Z$};

\node[left] at (0,0) {$\sigma$};

\draw[->, thick, purple] (2*\dx, 1.9) -- (12*\dx, 0.2);
\draw[->, thick, purple] (4*\dx, 1.9) -- (1*\dx, 0.2);
\draw[->, thick, purple] (5*\dx, 1.9) -- (6*\dx, 0.2);
\draw[->, thick, purple] (8*\dx, 1.9) -- (17*\dx, 0.2);

\draw[->, thick, purple] (13*\dx, 1.9) -- (3*\dx, 0.2);

\draw[
    decorate,
    decoration={brace, amplitude=6pt}
]
    (13.5*\dx, 2.55) -- (17.5*\dx, 2.55)
    node[midway, yshift=15pt] {$> 3d$};

\draw[
    decorate,
    decoration={brace, amplitude=6pt}
]
    (8.5*\dx, 2.55) -- (12.5*\dx, 2.55)
    node[midway, yshift=15pt] {$> 3d$};

    \draw[
    decorate,
    decoration={brace, amplitude=6pt}
]
    (18*\dx, 2.55) -- (23*\dx, 2.55)
    node[midway, yshift=15pt] {$> d + 1$};

\end{tikzpicture}
\caption{A permutation $\pi$ with a \emph{$d$-fitting} coloring with respect to a permutation $\sigma$, with \emph{witness} $\Sigma_{\pi} = \{1,2,3,4,5\}$.
Substrings highlighted in red form blocks.
The rightmost block contains more than $d$ red symbols, meaning that it cannot contain any moved symbol.}
\label{fig:dfitting}
\end{figure}

\begin{definition}[$d$-fitting coloring and witness] \label{def:fitting_coloring}
    Let $d\in \Nbb$, $\pi \in \Scal_{n}(\Sigma)$, and $\col \colon \Sigma \to \{\red, \blue\}$ be a coloring function of $\pi$.
    We say $\col$ is a \emph{$d$-fitting} coloring (of $\pi$) with respect to some permutation $\sigma \in \Scal_{n}(\Sigma)$, if there is a subset of symbols $\Sigma_{\pi}\subseteq \Sigma$ such that
    \begin{itemize}
        \item $|\Sigma_{\pi}|=\dist(\pi, \sigma)$ and $\sigma$ can be obtained from $\pi$ by moving the symbols in $\Sigma_{\pi}$;
        \item for every symbol in $x \in \Sigma_{\pi}$, we have $\col(x) = \red$; and
        \item for every symbol $x \in \Sigma \setminus \Sigma_{\pi}$, if $x$ appears in $\pi$ within $3d$ symbols (before or after) of some symbol in $\Sigma_{\pi}$, then $\col(x) = \blue$.
    \end{itemize}
    Moreover, such a subset of symbols $\Sigma_{\pi}$ is called a \emph{witness} of the coloring $\col$ to be $d$-fitting.
\end{definition} 

Intuitively, red symbols are candidates for being moved (i.e., they may or may not actually be moved), but blue symbols are definitely not moved.
Moreover, if a red symbol is indeed one of the moved symbols (i.e., it belongs to $\Sigma_{\pi}$), then all symbols within distance $3d$ of it must be correctly colored blue (if they are not moved) or red (if they are moved).
In other words, whenever the coloring correctly identifies a moved symbol as red, it must also correctly color its entire $3d$-neighborhood.

Consequently, the red symbols in a $d$-fitting coloring $\col$ of $\pi$ w.r.t.\ $\sigma$ that are not separated by $3d$ consecutive blue symbols are tied to each other in the sense that either both or none of them belong to any witness $\Sigma_{\pi}$.
This motivates the following definitions of blocks.

\begin{definition}[Blocks]
    Let $d\in \Nbb$, $\pi \in \Scal_{n}(\Sigma)$ and $\col \colon \Sigma \to \{\red,\blue\}$.
    A \emph{block} is a maximal substring (i.e., consecutive sequence) $\beta$ of $\pi$ such that
    \begin{itemize}
        \item the first and last symbols of $\beta$ are colored red; and
        \item $\beta$ contains no $3d$ consecutive blue symbols.
    \end{itemize}
\end{definition}

See also Figure~\ref{fig:dfitting} for a visualization of blocks.
By definition, any two blocks are disjoint and separated by at least $3d$ blue vertices in $\pi$.
We immediately obtain the following observation by recoloring all symbols in a block blue if it contains more than $d$ red symbols.

\begin{observation}\label{obs:blocks}
     Let $d \in \Nbb$, $\pi \in \Scal_{n}(\Sigma)$.
     Suppose $\col$ is a $d$-fitting coloring of $\pi$ w.r.t.\ some permutation $\sigma \in \Scal_{n}(\Sigma)$.
     If $\dist(\pi, \sigma) \le d$, there is a $d$-fitting coloring $\col'$ of $\pi$ w.r.t.\ $\sigma$ such that there are at most $d$ red symbols in each block.
     Moreover, such a coloring $\col'$ can be obtained from $\col$ in $\Ocal(n)$ time.
\end{observation}

We next argue how the notion of $d$-fitting colorings and that of witnesses provide a means to find a permutation that is one step closer to any true center $\sigma$ than the candidate $\pi_c$.

\begin{lemma}\label{lem:progress_from_witnesses}
    Let $d\in \Nbb$, $\pi_c, \pi_g \in \Scal_{n}(\Sigma)$.
    Let $\col_c$ and $\col_g$ be $d$-fitting colorings of $\pi_c$ and $\pi_g$ w.r.t.\ some permutation $\sigma \in \Scal_{n}(\Sigma)$ with $1 \le \dist(\pi_c, \sigma) \le d$ and $\dist(\pi_g, \sigma) \le d$.
    Consider any symbol $x\in \Sigma$ such that there are witnesses $\Sigma_c$ and $\Sigma_g$ of $\col_c$ and $\col_g$, respectively, with $x\in \Sigma_c\setminus \Sigma_g$. 
    Then in time $\Ocal(dn)$ we can compute and enumerate a set $P$ with $|P|\le 6d+2$ and there is $\tilde{\pi}\in P$ with $\dist(\tilde{\pi}, \sigma)=\dist(\pi_c, \sigma) -1$.
    The computation does not require the sets $\Sigma_c$ and $\Sigma_g$ as input.
\end{lemma}
\begin{proof}
    Consider an arbitrary pair $\Sigma_c$ and $\Sigma_g$ of witnesses.
    Intuitively, $x\in \Sigma_c$ means that by moving $x$ to a specific new location in $\pi_c$, we make $\pi_c$ closer to $\sigma$. As $x\notin \Sigma_g$, the position of $x$ in $\pi_g$ gives an indication for that location.
    
    To see where to move $x$ in $\pi_c$, let $b_l$ be the closest symbol to the left of $x$ in $\pi_g$ that is still blue in both permutations, and let $b_r$ be the closest such symbol to the right of $x$ in $\pi_g$. 
    If there is no such symbol on the right and/or left, consider $b_l$ or $b_r$ to be a \emph{phantom} blue symbol before/after the first/last symbol of the permutation, respectively.
    As $\Sigma_g$ does not contain any blue symbols, it does not contain any of $x, b_l,$ and $b_r$. Thus, as $x$ lies between $b_l$ and $b_r$ in $\pi_g$, it lies between them in $\sigma$ as well.
    As $b_l$ and $b_r$ are blue in $\col_c$ as well, they do not change their position between $\pi_c$ and $\sigma$.
    We add all permutations to $P$ that are obtained by moving $x$ in $\pi_c$ to some location between $b_l$ and $b_r$.
    This ensures that we consider all possible ways of moving $x$ closer to its position in $\sigma$, leading to the fact that $P$ always contains a permutation $\tilde{\pi}$ that is one step closer to $\sigma$ than $\pi_c$.

    We argue that we only need to add at most $6d+2$ such permutations to $P$.
    Consider the next $2d + 1$ symbols to the left of $x$ in $\pi_g$.
    Note that they belong to at most one block in $\pi_g$.
    In each block, there are at most $d$ red symbols, and at least $d+1$ of these symbols are blue in $\col_g$. Thereby, these symbols do not appear in any witness $\Sigma_g$ of $\col_g$. Let $L$ be the set of these blue symbols (select the closest $d+1$ to $x$, if there are more than $d+1$).
    As $\dist(\pi_g, \sigma) \le d$, the symbols in $L$ appear as a subsequence of $\sigma$ which is interrupted only by at most $d$ other symbols. Call the set of these interrupting symbols $I$.
    As the symbols in $L\cup I$ occur consecutively in $\sigma$, the symbols in $(L \cup I) \setminus \Sigma_c$ appear consecutively in $\pi_c$, except some potential interruptions by symbols in $\Sigma_c$. Thus, the symbols in $L\setminus \Sigma_c$ have a pairwise maximal distance of at most $|L|+|I|+|\Sigma_c| -1\le 3d$ in $\pi_c$. Thus, they belong to at most one block in $\pi_c$ and so at most $d$ of them are red in $\col_c$. Hence, there is a symbol $b_l\in L$ which is blue in both $\col_g$ and $\col_r$. 
    The same line of argument shows that a suitable symbol $b_r$ exists among the $2d+1$ symbols to the right of $x$ in $\pi_g$.
    
    As the distance between $b_l$ and $b_r$ is at most $4d+2$ in $\pi_g$ and $\dist(\pi_c, \pi_g) \le 2d$, the distance between $b_l$ and $b_r$ in $\pi_c$ is at most $6d+2$, giving that many possibilities on where to insert $x$.
    If one or both of $b_l$ and $b_r$ are phantom symbols at the beginning/end of a permutation, the same upper bound applies.

    We can find $b_\ell$ and $b_r$ in time $O(d)$ by starting a search outwards from $x$ in $\pi_c$. Listing all obtained permutations takes time in $\Ocal(dn)$.
\end{proof}

It remains to show how to find such a symbol $x$.
For technical reasons, we require a weak notion of consistency between two $d$-fitting colorings, which motivates the following definition.

\begin{definition}[Consistent $d$-fitting colorings] \label{def:consistent_coloring}
    Let $d\in \Nbb$, $\pi_c, \pi_g \in \Scal_{n}(\Sigma)$, and $\col_c, \col_g$ be $d$-fitting colorings of $\pi_c$ and $\pi_g$ with respect to some permutation $\sigma \in \Scal_{n}(\Sigma)$, respectively.
    We say that $\col_c$ and $\col_g$ are \emph{consistent $d$-fitting} with respect to $\sigma$ if there are respective witnesses $\Sigma_{c}, \Sigma_{g} \subseteq \Sigma$, such that for every symbol $x\in \Sigma_{c} \cup \Sigma_{g}$, we have 
    $\col_g(x) = \red$ if and only if $x \in \Sigma_{g}$, and $\col_c(x) = \red$ if and only if $x \in \Sigma_{c}$.
\end{definition} 

We first show how to enumerate a set of candidate pairs of colorings to be consistent $d$-fitting with respect to any $\sigma$.

\begin{lemma} \label{lem:fitting_candidates}
    Let $d \in \Nbb$ and $\pi_c, \pi_g \in \Scal_{n}(\Sigma)$.
    There is a set $\Ccal$ of $2$-coloring pairs such that 
    \begin{itemize}
        \item the size of $\Ccal$ is bounded by $|\Ccal| \le 2^{\Ocal(d^{2})} \log n$; and
        \item for every $\sigma \in \Scal_{n}(\Sigma)$ with $\dist(\pi_c, \sigma) \le d$ and $\dist(\pi_g, \sigma) \le d$, $\Ccal$ contains a consistent $d$-fitting coloring pair of $\pi_c$ and $\pi_g$ w.r.t.\ $\sigma$.
    \end{itemize}
    Moreover, such a set can be computed and listed in $2^{\Ocal(d^{2})} n\log n$ time, where $n \coloneqq |\Sigma|$. 
\end{lemma}

\begin{proof}
    We construct the collection $\Ccal$ using the notion of $(a, b)$-universality.
    An $(a, b)$-universal set $\Ucal \subseteq \{0,1\}^{a}$ is a set of bitstrings of length $a$, such that for each subset $\Scal \subseteq [a]$ of size $b$, the projection of $\Ucal$ onto the coordinates in $\Scal$ yields all $2^b$ possible bitstrings of length $b$.
    Equivalently, if we restrict each string in $\Ucal$ to any fixed set of $b$ coordinates, all possible $2^b$ binary patterns appear.
    Naor, Schulman, and Srinivasan~\cite{NaorSS95} showed that such a set exists with size at most $2^b b^{\Ocal(\log b)} \log a$, and can be listed in time $2^b b^{\Ocal(\log b)} a \log a$.

    In our setting, we take $a = 2n$ and $b = 12d^2 + 4d$.
    Applying the above, we obtain an $(a, b)$-universal set $\Ccal \subseteq \{0,1\}^{a}$ of size
    \begin{equation*}
        2^b b^{\Ocal(\log b)} \log (a) = 2^{\Ocal(d^{2} + \log^{2} d)} \log n = 2^{\Ocal(d^{2})} \log n,
    \end{equation*}
    which can be listed it in time $2^{\Ocal(d^{2})} n \log n$.
    Each bitstring in $\{0,1\}^{2n}$ naturally defines a pair of $2$-colorings of the symbols in $\Sigma$: the first $n$ bits define the coloring $\col_c$ of $\pi_c$ and the last $n$ bits define a coloring $\col_g$ of $\pi_g$.

    Now, consider any permutation $\sigma$ with $\dist(\pi_c, \sigma) \le d$ and $\dist(\pi_g, \sigma) \le d$.
    Let $\Sigma_{c}$ and $\Sigma_{g}$ be minimum-sized sets of moved symbols in $\pi_c$ and $\pi_g$, respectively, to obtain~$\sigma$.
    By Definitions~\ref{def:fitting_coloring} and~\ref{def:consistent_coloring}, the $2$-colorings $\col_c$ and $\col_g$ are consistent $d$-fitting w.r.t.\ $\sigma$ if
    \begin{enumerate}[(i)]
        \item every symbol in $\Sigma_{c}$ (resp., $\Sigma_{g}$) is colored red in $\col_c$ (resp., $\col_g$);
        \item every symbol in $\Sigma \setminus \Sigma_{c}$ (resp., $\Sigma \setminus \Sigma_{g}$) within distance $3d$ of a moved symbol in $\Sigma_{c}$ (resp., $\Sigma_{g}$) is colored blue in $\col_c$ (resp., $\col_g$); and
        \item every symbol in $\Sigma_{g} \setminus \Sigma_{c}$ (resp., $\Sigma_{c} \setminus \Sigma_{g}$) is colored blue in $\col_g$ (resp., $\col_c$).
    \end{enumerate}
    
    Recall that both sets have size at most $d$.
    The number of symbols involved in these conditions is at most 
    \begin{equation*}
        (|\Sigma_{c}| + |\Sigma_{g}|) + 2 \cdot 3d(|\Sigma_{c}| + |\Sigma_{g}|) + (|\Sigma_{c}| + |\Sigma_{g}|) \leq 12d^{2} + 4d = b.
    \end{equation*}
    Because our $(a, b)$-universal set contains all possible pairs of colorings on any set of $b$ positions, it must contain a pair that assigns suitable colors to all $b$ involved symbols.
    Hence, this pair is consistent $d$-fitting w.r.t.\ $\sigma$.
\end{proof}

Given a guide $\pi_g$ and the candidate $\pi_c$, \Cref{lem:fitting_candidates} allows us to guess a consistent $d$-fitting coloring pair $(\col_c, \col_g)$ w.r.t.\ some center $\sigma$ within distance $d$ of both.
We now use the colorings $\col_c$ and $\col_g$ to capture the structural relationship between the two permutations.
To this end, we consider a certain vertex cover of their \emph{permutation graph}.
Here, a permutation graph $G_{\pi_c, \pi_g}$ of two permutations $\pi_c, \pi_g \in \Scal_{n}(\Sigma)$ is the graph whose vertex set is the alphabet $\Sigma$ and its edges are exactly the (unordered) pairs of vertices whose order is reversed between $\pi_c$ and $\pi_g$.
For each edge in $G_{\pi_c, \pi_g}$, exactly one of $\pi_c$ and $\pi_g$ has the two incident symbols in the same order as the center $\sigma$.
A vertex cover of $G_{\pi_c, \pi_g}$ thus gives an indication of which symbols need to be moved to transform $\pi_g$ and $\pi_c$ into $\sigma$. 

\begin{lemma}\label{lem:exists_vertexcover_crossing}
    Let $d\in \Nbb$, $\pi_c, \pi_g \in \Scal_{n}(\Sigma)$.
    Let $\col_c$ and $\col_g$ be consistent $d$-fitting colorings of $\pi_c$ and $\pi_g$ w.r.t.\ some permutation $\sigma \in \Scal_{n}(\Sigma)$.
    For every pair of witnesses $\Sigma_c, \Sigma_g$ for $\col_c$ and $\col_g$ there is an inclusion-wise minimal vertex cover $X$ of $G_{\pi_c, \pi_g}$ with $X\subseteq \Sigma_c \cup \Sigma_g$.
\end{lemma}
\begin{proof}
    We begin by showing that $\Sigma_{c} \cup \Sigma_{g}$ is a vertex cover of $G_{\pi_c, \pi_g}$.
    Consider an edge $xy$ in $G_{\pi_c, \pi_g}$.
    By definition of the permutation graphs, $x$ and $y$ appear in reversed order in $\pi_c$ and $\pi_g$.
    Since $\sigma$ is obtained from $\pi_{c}$ (resp., $\pi_{g}$) by moving symbols in $\Sigma_{c}$ (resp., $\Sigma_{g}$), at least one of $x$ and $y$ must be moved in one of the two permutations.
    Hence, at least one of $x$ and $y$ belongs to $\Sigma_{c} \cup \Sigma_{g}$, showing that $\Sigma_{c} \cup \Sigma_{g}$ is a vertex cover of $G_{\pi_c, \pi_g}$.

    The lemma follows by exhaustively removing elements from $\Sigma_{c} \cup \Sigma_{g}$ upon whose removal all edges are still covered.
\end{proof}

Such a vertex cover indeed suffices to find a suitable vertex $x$ as required for Lemma~\ref{lem:progress_from_witnesses}.

\begin{lemma}\label{lem:findXfromVertexCover}
     Let $d\in \Nbb$, $\pi_c, \pi_g \in \Scal_{n}(\Sigma)$.
    Let $\col_c$ and $\col_g$ be consistent $d$-fitting colorings of $\pi_c$ and $\pi_g$ w.r.t.\ some permutation $\sigma \in \Scal_{n}(\Sigma)$ with $\dist(\pi_c, \sigma) \le d$ and $\dist(\pi_g, \sigma) \le d$.
    Consider an inclusion-wise minimal vertex cover $X$ of $G_{\pi_c, \pi_g}$ with $X\subseteq \Sigma_c \cup \Sigma_g$ for a pair of witnesses $\Sigma_c, \Sigma_g$ for $\col_c$ and $\col_g$.
    Then in time $\Ocal(|X|)$ we can find a symbol $x\in \Sigma$ with $x\in \Sigma_c\setminus \Sigma_g$. The computation does not require the sets $\Sigma_c$ and $\Sigma_g$ as input.
\end{lemma}
\begin{proof}
    Observe that $\pi_g$ can be transformed into $\pi_c$ by deleting and reinserting the symbols in $X$, so $\dist(\pi_g, \pi_c) > d$ implies that $|X| > d$.  
    As $\dist(\pi_g, \sigma) = |\Sigma_g| \le d$, we have that there is a symbol $x\in X$ such that $x\in \Sigma_c \setminus \Sigma_g$.
    As $\col_g$ and $\col_c$ are consistent $d$-fitting, we have that $\col_c(x)=\red$ and $\col_g(x)=\blue$.
    It suffices to consider any such colored symbol $x'$, as $\col_g(x')=\blue$ implies $x'\notin \Sigma_g$ and thus, by $X\subseteq \Sigma_c \cup \Sigma_g$ we have $x'\in \Sigma_c\setminus \Sigma_g$.
\end{proof}
With this, we are ready to forge our central tool:

\begin{lemma}
    \label{lem:center_dk_fpt_guide}
    Let $d\in \Nbb$, and $\pi_g, \pi_c \in \Scal_{n}(\Sigma)$ be two permutations with $d < \dist(\pi_g, \pi_c) \leq 2d$. 
    There exists a set of permutations $P \subseteq \Scal_{n}(\Sigma)$ satisfying
    \begin{itemize}
        \item the size of $P$ is bounded by $|P| \le 2^{\Ocal(d^{2})} \log n$; and
        \item for each permutation $\sigma \in \Scal_{n}(\Sigma)$ with $\dist(\pi_g, \sigma) \leq d$ and $\dist(\pi_c, \sigma) \leq d$, $P$ contains a permutation $\tilde{\pi} \in P$ that is one step closer to $\sigma$ than $\pi_{c}$, i.e., $\dist(\tilde{\pi}, \sigma) = \dist(\pi_c, \sigma) -1$.
    \end{itemize}
    Furthermore, such a set can be computed in $2^{\Ocal(d^{2})} n \log n$ time.
\end{lemma}

\begin{proof}
    We compute a collection $\Ccal$ of pairs of $2$-colorings of $\Sigma$ according to Lemma~\ref{lem:fitting_candidates}.
    For each $\sigma \in \Scal_{n}(\Sigma)$ with $\dist(\pi_g, \sigma) \leq d$ and $\dist(\pi_c, \sigma) \leq d$, there exists a consistent $d$-fitting pair of colorings $(\col_c,\col_g) \in \Ccal$ w.r.t.\ $\sigma$.
    Thus, for each pair in $\Ccal$ we perform the following procedure, adding permutations to $P$ as we go along.

    Based on Observation~\ref{obs:blocks}, we can assume that both $\col_c$ and $\col_g$ have at most $d$ red symbols in each block (if not, we recolor some blocks blue accordingly).
    We list all inclusion-wise minimal vertex covers of $G_{\pi_c, \pi_g}$ that have size at most $2d$. 
    We can do so by a simple and well-known branching technique: Start with an empty cover, iteratively take any uncovered edge $vw$ and create two branches where you add $v$ or $w$ to the vertex cover, respectively, and repeat this procedure $2d$ times (or until all edges are covered). This list of at most $2^{2d}=4^d$ covers may contain some that are not inclusion-wise minimal, but these can be filtered out in time  $\Ocal(4^d \cdot |V(G_{\pi_c, \pi_g})|)=\Ocal(4^d n)$.

    By Lemma~\ref{lem:exists_vertexcover_crossing} and as every witness for $\col_c$ or $\col_g$ has size at most $d$, at least one of the listed vertex covers consists only of vertices in $\Sigma_c \cup \Sigma_g$, where $\Sigma_c$ and $\Sigma_g$ are witnesses for the respective colorings. This holds true for \emph{every} permutation $\sigma$ for which the current colorings are consistent $d$-fitting.
    Using Lemma~\ref{lem:findXfromVertexCover} on each of the vertex covers, we obtain a set of at most $4^d$ symbols, and for every suitable $\sigma$ at least one of these symbols is included in $\Sigma_c \setminus \Sigma_g$ for respective witnesses. 
    We note that applying Lemma~\ref{lem:findXfromVertexCover} to vertex covers containing vertices which do not stem from union of witnesses might yield incorrect symbols, but for each of the vertex covers we still add at most one symbol.
    We apply Lemma~\ref{lem:progress_from_witnesses} to each of the $4^d$ symbols, thereby producing a set of $4^d \cdot (6d+2)$ permutations, which we all add to $P$. 
    Note that the lemmas guarantee that for every $\sigma$ for which the colorings are consistent $d$-fitting, we added at least one permutation which is one step closer to $\sigma$ than the current candidate $\pi_c$.

    The final size of $P$ is bounded by
    \begin{equation*}
        |P| \le |\Ccal| \cdot 4^d (6d + 2) \le 2^{\Ocal(d^{2})} \log n.
    \end{equation*}

    The runtime estimate is based on combining the size of $\Ccal$ from Lemma~\ref{lem:fitting_candidates} with the time to find the unique vertex cover respecting a given pair of colorings.
    All other steps, like processing the blocks, are asymptotically dominated by that time.
    There are at most $2^{\Ocal(d^{2})} \log n$ pairs of colorings to enumerate.
    For each pair, the branching procedure on the vertex covers takes 
    $
        \Ocal\bigl(
                4^d \cdot (|E(G_{\pi_c, \pi_g})| + |V(G_{\pi_c, \pi_g})| ) 
             \bigr)
        = \Ocal(4^d dn)
    $ 
    time as $|V(G_{\pi_c, \pi_g})|=n$ and the existence of a vertex cover of size $2d$ implies that $|E(G_{\pi_c, \pi_g})|\le 2dn$.
    Thus, the total time to compute $P$ is in 
    \begin{equation*}
        2^{\Ocal(d^{2})} \log n \cdot \Ocal(4^{d} dn) = 2^{\Ocal(d^{2})} n \log n.\qedhere
    \end{equation*}
\end{proof}

Now we are ready to prove \Cref{thm:center_dk_fpt}. 

\centerFPTdk*

\begin{proof}
    We may assume without loss of generality that any \yesinstance admits a solution with exactly $k$ centers.
    Indeed, if a solution exists with fewer than $k$ centers, we can always add arbitrary permutations as centers until we reach exactly $k$ centers.
    In addition, we exclude the trivial \yesinstances where $k > m$.
    Now, consider a hypothetical solution $S^{*} = \{\sigma_1^{*}, \sigma_2^{*}, \ldots, \sigma_k^{*}\}$ consisting of $k$ centers.
    For each $i \in [k]$, let the $i$th cluster be the set of permutations in $\Pi$ with distance at most $d$ to $\sigma_{i}^{*}$.
    If a permutation would fit into multiple clusters, we assign it to only one of them arbitrarily. 

    We describe our algorithm as a branching procedure.
    It maintains a candidate set $S \subseteq \Scal_{n}(\Sigma)$ of size at most $k$ along with a budget $b_i$ for each candidate $\sigma_i \in S$.
    Here, $S$ is intended to evolve into a solution to the \ucenterfull instance, and each budget $b_i$ indicates how many modifications to $\sigma_i$ remain (or have been used so far).
    In each branch, the algorithm checks whether the current candidate set $S$ covers all permutations in $\Pi$ within distance $d$.
    The algorithm halts and accepts the instance if and only if at least one branch finds a valid solution.
    We now describe the procedure in detail; see Algorithm~\ref{alg:fpt_center_dk} for a high-level overview.

    \begin{algorithm}[t!]
    \SetAlgoLined
    \caption{Solving \ucenterfull for radius $d$ on permutations $\Pi$.}\label{alg:fpt_center_dk}
    \KwIn{An alphabet $\Sigma$ of size $n$, a set of permutations $\Pi$ if size $m$, integers $k$ and $d$.}
    \KwOut{Decide the whether there exist $S \subseteq \Scal_{n}(\Sigma)$ of size $k$ such that for each $\pi \in \Pi$, there exists $\sigma \in S$ with $\dist(\pi, \sigma) \le d$.}
    
    $\sigma_{1} \gets \pi_{1} \in \Pi$; \tcp{Initialize the first candidate center}
    $b_{1} \gets d$; \tcp{Initialize the budget of $\sigma_{1}$}
    \Return $\textsc{Alg}\big(\Sigma, \Pi, k, d, \{\sigma_{1}\}, \{b_{1}\}\big)$\;
    \SetKwProg{Fn}{Procedure}{}{}
    \Fn{$\textsc{Alg}\big(\Sigma$, $\Pi$, $k$, $d$, $S=\{\sigma_1, \ldots, \sigma_{|S|}\}$, $\{b_i\}_{i = 1}^{|S|}\big)$}{
        \lIf{for every $\pi \in \Pi$ there is $\sigma_{i} \in S$ such that $\dist(\pi, \sigma_{i}) \le d$}{\textbf{accept instance}}
        Pick arbitrary $\pi \in \Pi$ with $\dist(\pi, \sigma_{i}) > d$ for all $\sigma_{i} \in S$\;
        \tcp{Branching on how a permutation not yet covered by $S$ is covered}
        \For {each $\sigma_{i} \in \{\sigma_i \in S : b_i > 0, \dist(\pi, \sigma_i) \le 2d\}$}
        {
            \tcp{The case that $\pi$ is in the $i$th cluster.}
            $P \gets $ invoke~\cref{lem:center_dk_fpt_guide} with $(d, \pi_g, \pi_c)$, where $\pi_g \gets \pi$ and $\pi_c \gets \sigma_{i}$\;
            \For {each $\sigma_{i}'\in P$}
            {
                $\sigma_{i} \gets \sigma_{i}'$; \tcp{Modify the candidate $\sigma_{i}$}
                $b_{i} \gets b_{i} -1$; \tcp{Update the budget of $\sigma_{i}$}
                 \textbf{accept instance} if $\textsc{Alg}\big(\Sigma, \Pi, k, d, S, \{b_i\}_{i = 1}^{|S|}\big)$ accepts\;
            }
        }
        \If {$|S| < k$}
        {   \tcp{The case that $\pi$ belongs to a new cluster.}
            $\sigma_{j} \gets \pi$, where $j = |S| + 1$; \tcp{Add $\pi$ as a new candidate center}
            $b_{j} \gets d$; \tcp{Initialize the budget of $\pi$.}
            \textbf{accept instance} if $\textsc{Alg}\big(\Sigma, \Pi, k, d, S \cup \{ \sigma_{j} \}, \{b_i\}_{i = 1}^{j}\big)$ accepts; 
        }
        \textbf{reject branch};
    }
    \end{algorithm}

    In the course of the algorithm, we one by one initialize each candidate $\sigma_i$ as a permutation from $\Pi$, and evolve it into the $i$th center.
    Initially, pick the permutation $\pi_{1} \in \Pi$ and set the first candidate $\sigma_1$ to equal this permutation: $\sigma_1 \coloneqq \pi$.
    Then, the algorithm adds it as the first candidate to $S$ and initializes its budget as $b_1 \coloneqq d$.
    
    Next, the algorithm branches on a permutation $\pi \in \Pi$ that has distance more than $d$ to each of the candidates in $S$. 
    If there is none, then the current candidates form a solution, and we accept the instance.
    Otherwise, the permutation $\pi$ either belongs to one of the existing clusters or to a new cluster.
    \begin{itemize}
        \item \textbf{Existing cluster:} For each candidate $\sigma_i \in S$ with remaining budget $b_i > 0$ and distance at most $2d$ to $\pi$, we consider the case where $\pi$ belongs to the $i$th cluster.
        We invoke \cref{lem:center_dk_fpt_guide} with $d$, guide $\pi_g = \pi$, and $\pi_c = \sigma_i$ to obtain the set $P$ of permutations.
        For each permutation $\sigma_i' \in P$, we create a new branch where we replace $\sigma_i$ by $\sigma_i'$ and decrease its budget by one.
        \item \textbf{New cluster:} If $|S| < k$, we create a branch where $\pi$ belongs to a new cluster.
        Specifically, we let $\sigma_{|S| + 1} = \pi$ as a new candidate, initialize its budget as $b_{|S| + 1} = d$, and add it to $S$.
    \end{itemize}
    If none of the two above cases applies, we reject the branch. We accept the instance if and only if at least one branch accepts it.
    If none of the branches accepts the instance, we reject the branch.

    Observe that if the algorithm accepts an instance, it indeed finds a solution.
    Conversely, consider a \yesinstance of \ucenterfull witnessed by some solution $S^{*} = \{\sigma_1^{*}, \sigma_2^{*}, \ldots, \sigma_k^{*}\}$.
    When branching on a permutation $\pi \in \Pi$ not covered by the current candidates $S$, there exists a branch where $\pi$ is assigned to the cluster to which it belongs in $S^{*}$.
    Moreover, when $\pi$ is assigned to the $i$th cluster from the current candidates $S$, by \cref{lem:center_dk_fpt_guide}, there is a permutation $\sigma_i' \in P$ that is one step closer to $\sigma_i^{*}$ than $\sigma_i$.
    Thus, in the respective branch, we replace $\sigma_i$ by $\sigma_i'$, moving it one step closer to $\sigma_i^{*}$.
    Since each candidate starts with the distance at most $d$ to its respective center in $S^{*}$, after receiving at most $d$ updates, it will have converged to that center.
    From this point on, it will not receive any more updates since there are no permutations assigned to that cluster with a distance of more than $d$.
    Therefore, there is a branch in which the algorithm finds the solution $S^{*}$ and accepts the instance.

    Finally, we analyze the running time of the algorithm.
    Recall that the input of the branching procedure contains a set of candidates $S$ along with their budgets $\{ b_{i} : \sigma_i \in S \}$.
    Define the measure of the current input of a branch as
    \begin{equation*}
        \mu \coloneqq (k - |S|)(d + 1) + \sum_{i=1}^{|S|} b_i.
    \end{equation*}
    and denote by $T(\mu)$ the worst-case running time of the algorithm on inputs with measure $\mu$.

    Observe that each time we add a new candidate to $S$ or modify an existing candidate in $S$, the measure $\mu$ decreases by one.
    On the one hand, when we add a new candidate to $S$, the size of $S$ increases by one while the budget of the new candidate is initialized to $d$.
    Thus, the measure decreases by $(d + 1) - d = 1$ and there is at most one such branch.
    On the other hand, when we modify a candidate to a permutation from $P$, its budget decreases by one while the size of $S$ remains unchanged.
    Thus, the measure decreases by $1$ as well, and there are at most $|S| \cdot |P| \leq k \cdot 2^{\Ocal(d^{2})} \log n$ such branches by \cref{lem:center_dk_fpt_guide}.
    Moreover, computing $P$ takes time in $2^{\Ocal(d^{2})} n \log n$, and the remaining time spend on this iteration is dominated by the computation of all Ulam distances to the current candidates (achievable in $\Ocal(kmn \log n)$ time).   
    Therefore, we have the recurrence
    \begin{equation*}
        T(\mu) \leq 2^{\Ocal(d^{2})} n \log n + \Ocal(kmn \log n) + \big(1 + k \cdot 2^{\Ocal(d^{2})} \log n\big) \cdot T(\mu - 1).
    \end{equation*}
    Notice that the measure satisfies $\mu \leq k(d + 1)$ and the algorithm halts when $\mu = 0$.
    Solving the recurrence yields
    \begin{equation*}
        \begin{aligned}
            T(k(d + 1)) 
            ={}& \big(1 + k \cdot 2^{\Ocal(d^{2})} \log n\big)^{k(d + 1)} \cdot (2^{\Ocal(d^{2})} n \log n + kmn \log n) \\
            ={}& 2^{\Ocal(d^{3}k)} \cdot k^{\Ocal(dk)} \cdot (\log n)^{k(d+1)} \cdot \Ocal(mn\log n).
        \end{aligned}
    \end{equation*}
    Observe that
    \begin{equation*}
        (\log n)^{k(d+1)} 
        = 2^{k(d+1)\log\log n} 
        \le 2^{k^2(d+1)^2 + (\log\log n)^2} 
        = 2^{k^2(d+1)^2} \cdot n^{(\log \log n)^{2} / \log n},
    \end{equation*}
    which bounds the total running time by
    \begin{equation*}
        \begin{aligned}
        &2^{\Ocal(d^{3} k + d^{2}k^{2})} \cdot k^{\Ocal(dk)} \cdot \Ocal\big(mn^{1 + (\log \log n)^{2} / \log n} \log n\big) 
        = 2^{\Ocal(d^{2} k(d + k))}  \cdot mn^{1 + o(1)}.
        \end{aligned}
    \end{equation*}
    This shows that \ucenterfull is in $\FPT$ by $k+d$. \qedhere

\end{proof}

\subsection{Lower Bound on Kernelization by \texorpdfstring{$k+d$}{k+d}}\label{sec:nokernel}

We next prove that the fixed-parameter tractability of \ulamdistance with respect to $k + d$ does not extend to a polynomial kernel.
To this end, we provide a reduction from a problem on binary strings.

A crucial tool is a \emph{distance-preserving} translation $f \colon \{0,1\}^\ell \to \Scal_{2\ell}([2\ell])$ between binary strings and permutations such that the Hamming distances between two strings equal the Ulam distances between their corresponding permutations.
Let a binary string be given as $s = s_1 s_2 \cdots s_\ell$ with $s_i \in \{0,1\}$.
Intuitively, we employ two symbols per bit; the order of these symbols in the permutation then reflects the value of that bit.
Formally, we define $f(s) = \highlight{x_1\ y_1}\ \highlight{x_2\ y_2} \cdots \highlight{x_\ell\ y_\ell},$
where
\begin{equation*}
    (x_i, y_i) =
    \begin{cases}
        (2i-1,\; 2i), & \text{if } s_i = 0; \\[6pt]
        (2i,\; 2i-1), & \text{if } s_i = 1.
    \end{cases}  
\end{equation*}
For example, $f(00101) = \highlight{1\ 2}\ \highlight{3\ 4}\ \highlight{6\ 5}\ \highlight{7\ 8}\ \highlight{10\ 9}$. 

It is straightforward to see that $f(\cdot)$ is a bijection between the set of binary strings of length $\ell$ and permutations over $\Sigma = [2\ell]$ respecting the \emph{pair scheme}.
We say that a permutation $\pi$ over $[2\ell]$ \textit{respects the pair scheme} if and only if it can be written as
\begin{equation*}
    \pi = \highlight{x_1 \ y_1} \ \highlight{x_2 \ y_2} \cdots \highlight{x_\ell \ y_\ell},
\end{equation*} 
where for each $i \in [\ell]$, $(x_i, y_i) = (2i-1, 2i)$ or $(2i, 2i-1)$.
We now argue that $f$ is distance-preserving.

\begin{lemma}
    \label{lem:map_from_bin_strings}
    Let $s, t \in \{0,1\}^\ell$ be binary strings and $f$ be as defined above.
    Then the Hamming distance between $s$ and $t$ satisfies
    $\distH (s, t) = \dist(f(s), f(t))$.
\end{lemma}

\begin{proof}
    We first show that $\distH (s, t) \geq \dist(f(s), f(t))$.
    Recall that the Ulam distance between the two permutations $f(s)$ and $f(t)$ is defined as the length of the permutations minus the length $\LCS(f(s),f(t))$ of a longest common subsequence.
    To construct a common subsequence between $f(s)$ and $f(t)$, observe that $1\ 3\ 5 \cdots 2\ell-1$ is a common subsequence of $f(s)$ and $f(t)$.
    Next, consider each index $i \in [\ell]$ such that the $s_i = t_i$.
    The common subsequence can be extended by inserting the symbol $2i$ right before (if $s_i=t_i=1$) or after (if $s_i=t_i=0$) the symbol $2i-1$.
    Thus, we have 
    \begin{equation*}
        \dist(f(s), f(t)) = |\Sigma| - \LCS(f(s), f(t)) \leq 2\ell - (\ell + \ell - \distH (s,t)) = \distH (s,t).
    \end{equation*}
    
    For the other direction, observe that for each index $i \in [\ell]$ such that $s_i \neq t_i$, any longest common subsequence of $f(s)$ and $f(t)$ can only pick either $2i-1$ or $2i$.
    Thus, we have 
    \begin{equation*}
        \LCS(f(s),f(t)) 
        \leq 2\ell - \distH (s,t),
    \end{equation*}
    and
    \begin{equation*}
        \dist(f(s),f(t)) 
        = 2\ell - \LCS(f(s),f(t)) \geq \distH (s,t).
    \end{equation*}
    Combining both inequalities, we have the desired statement.
\end{proof}

 We are now ready to state our reduction from the \cstringfull problem. %

  \defproblema{\cstringfull}
    {An alphabet $\Sigma$ of size $n$, a set $\Scal = \{s_1,s_2,\ldots, s_m\}$ of strings of length $\ell$ over $\Sigma$, and an integer $d$.}
    {Does there exists a string $s$ length $\ell$ such that $\distH (s, s_i) \leq d$ for all $i \in [m]$,
    where $\distH $ denotes the Hamming distance?
    \vspace{2mm}}

\begin{lemma}
    \label{lem:red_string_to_center}
    There is a polynomial parameterized transformation (PPT) from \cstringfull to $\ucenterfull$, for parameter $d$ in each problem.
\end{lemma}

\begin{proof} 
    Let $(\Scal, d)$ be an instance of \cstringfull, where $\mathcal{S} = \{s_1, s_2, \ldots, s_m\} \subseteq \{0,1\}^\ell$.
    We construct an instance of \ucenterfull $(\Pi,k',d')$ as $\Pi = \{f(s_i) : s_i \in S\}$, $k'= 1$, and $d' = d$.
    Here, $f$ is the distance-preserving function defined above.

    Suppose $(\mathcal{S}, d)$ is a \yesinstance of \cstringfull, and let $s \in \{0,1\}^\ell$ be a string such that $\distH(s, s_i) \leq d$ for all $s_i \in \mathcal{S}$.
    By \Cref{lem:map_from_bin_strings}, we have $\dist(f(s), f(s_i)) = \distH(s, s_i) \leq d$ for all $i \in [m]$, so $\{f(s)\}$ is a solution to the \ucenterfull instance $(\Pi,k',d')$.

    For the other direction, assume $(\Pi,k',d)$ is a \yesinstance of $\ucenterfull$, witnessed by some center permutation $\pi_c$ such that $\dist(\pi_c,\pi) \leq d$ for all $\pi \in \Pi$.
    
    Next, consider the structure of the center permutation $\pi_c$.
    If $\pi_c$ already respects the pair scheme, then by \Cref{lem:map_from_bin_strings}, the string $f^{-1}(\pi_c)$ is a solution to the instance $(\Scal, d)$ of $\cstringfull$, and we are done.
    Otherwise, we construct a new permutation $\pi'_{c}$ that respects the pair scheme as follows.
    For each $i \in [\ell]$, if $2i-1$ precedes $2i$ in $\pi_c$, then we set the $i$th pair of $\pi'_{c}$ as $\highlight{2i{-}1\ 2i}$; otherwise, set it as $\highlight{2i\ 2i{-}1}$.
    The only thing left is to show that $\pi'_{c}$ is also a solution to the \ucenterfull instance $(\Pi, k', d')$.

    Define $\rho(\pi, \pi_c)$ as the number of indices $i \in [\ell]$ for which the relative order of $2i-1$ and $2i$ differs between $\pi$ and $\pi_c$. 
    Formally,
    \begin{equation*}
        \rho(\pi, \pi_c) \coloneq \big|\big\{i \in [n] : (\pi^{-1}(2i-1)-\pi_c^{-1}(2i)) \cdot (\pi^{-1}(2i-1)-\pi_c^{-1}(2i)) <0 \big\} \big|.
    \end{equation*}
    Notice that any common subsequence of $\pi \in \Pi$ and $\pi_c$ can contain at most one element from each disagreeing pair $(2i-1, 2i)$ with $i \in [\ell]$.
    By the definition of Ulam distance, we have
    \begin{equation*}
        \dist(\pi_c, \pi) \geq \rho(\pi, \pi_c).
    \end{equation*}
    Additionally, by the construction of $\pi'_{c}$, the two permutations $\pi_c$ and $\pi'_{c}$ agree on the order within every pair $(2i-1, 2i)$ with $i \in [\ell]$.
    Hence, for any $\pi \in \Pi$, it holds that
    \begin{equation*}
        \dist(\pi'_{c}, \pi) = \rho(\pi, \pi'_{c}).
    \end{equation*}
    Combining these facts, we obtain
    \begin{equation*}
        \dist(\pi'_{c}, \pi) \leq \dist(\pi_c, \pi) \leq d,
    \end{equation*}
    for all $\pi \in \Pi$.

    Finally, since $\pi'_{c}$ respects the pair scheme, \Cref{lem:map_from_bin_strings} directly implies that the string $f^{-1}(\pi'_{c})$ is a solution to the instance $(\Scal, d)$ of \cstringfull.
    The reduction runs in polynomial time and preserves $d' = d$, yielding a PPT from \cstringfull to \ucenterfull with $k' = 1$.
    This completes the proof.
\end{proof}

Under well-established complexity assumptions, \cstringfull does not admit a polynomial-sized kernel with respect to $d$ \cite{BASAVARAJU201821}. For $\ucenterfull$ with $k = 1$ the same holds by \Cref{lem:red_string_to_center}.
By extension, this implies that there is no polynomial-sized kernel for the parameter $k+d$ for arbitrary instances of \ucenterfull.

\centernokernel*

\section{The Complexity of \texorpdfstring{\umedianfull}{Ulam Metric k-Median}}\label{sec:median}
    This section studies the parameterized complexity of \umedianfull with respect to the parameters $d$ and $k + d$.
We first prove that \umedianfull is $\W[1]$-hard parameterized by $d$ but is in $\XP$ for this parameterization.
Then we show a polynomial kernel parameterized by both $k$ and $d$, implying fixed-parameter tractability for this parameterization.

\subsection{Parameterization by \texorpdfstring{$d$}{d}}

In this subsection, we first rule out fixed-parameter tractability by $d$ alone under established complexity assumptions.
We then contrast it with a simple $\XP$-algorithm for \umedianfull.

For the hardness result, we give a reduction from 
    \textsc{Multicolored Clique}.
    An instance of \textsc{Multicolored Clique} consists of a graph $G$ with a proper vertex coloring using $k'$ colors, and asks whether $G$ contains a clique of size $k'$.
    The problem is known to be $\W[1]$-hard when parameterized by $k'$~\cite{Pietrzak03}.
    For clarity, we here describe the reduction and then separately prove that it preserves \yesinstances and \noinstances.

    Let $\Ccal = \{ c_{1}, c_{2}, \ldots, c_{k'}\}$ denote the color set, let $V(G)  = \{v_{1}, v_{2}, \ldots, v_{n'}\}$ be the vertex set of $G$ and let $E(G) = \{e_{1}, e_{2}, \ldots, e_{m'}\}$ be the edge set of $G$, where $n' = |V(G)|$ and $m' = |E(G)|$.
    Without loss of generality, we assume that $k' \geq 4$, as instances with smaller values can be solved in polynomial time by brute-force.
    
    We construct an instance $(\Pi, k, d)$ of \umedianfull with $k= m' - \binom{k'}{2} + 1$.
    The idea of the reduction is to create one permutation for each edge in $G$.
    We build the instance in a way such that the existence of a multicolored clique $Q \subseteq V(G)$ allows the construction of a center permutation $\sigma_{Q}$ such that edges with both endpoints in $Q$ have small total distance to $\sigma_{Q}$, while every other edge can take itself as its center.
    Consequently, the number of centers required is $m' - \binom{|Q|}{2} + 1 = k$, and the median distance is entirely contributed by the edges inside $Q$.
    It remains to choose a suitable value for $d$.

    Formally, let $q = \binom{k'}{2}(k'-2)+1$. Furthermore, we define the alphabet as the union of the color set, some $v$-symbols representing the vertices, and some $x$-, $y$-, and $z$-symbols representing the edges:
    \begin{equation*}
        \Sigma = 
        \Ccal 
        \cup \big\{v_{j}^{i} : i \in[q], v_j \in V(G) \big\}
        \cup \big\{x_{j}^{i}, y_{j}^{i}, z_{j}^{i} : i\in[q], e_{j} \in E(G) \big\}.
    \end{equation*}
    We have $n = |\Sigma| = k' + n'q + 3m'q$.
    We first construct a base permutation $\pi_{\base} \in S_{n}(\Sigma)$ that starts with alternating between the $x$- and $y$-symbols, followed by the $v$- and the $z$-symbols, and finally the color-symbols:
    \begin{align*}
        \pi_{\base} 
        \coloneq 
        {}& \highlight{x_{1}^{1}\;y_{1}^{1} \cdots x_{1}^{q}\;y_{1}^{q}}\
        \highlight{x_{2}^{1}\;y_{2}^{1} \cdots x_{2}^{q}\;y_{2}^{q}}
        \cdots
        \highlight{x_{m'}^{1}\;y_{m'}^{1} \cdots x_{m'}^{q}\;y_{m'}^{q}}\;\\
        {}&\highlight{v_{1}^{1}\;v_{1}^{2} \cdots v_{1}^{q}}\;
        \highlight{v_{2}^{1}\;v_{2}^{2} \cdots v_{2}^{q}}
        \cdots
        \highlight{v_{n'}^{1}\;v_{n'}^{2} \cdots v_{n'}^{q}}\;\\
        {}&\highlight{z_{1}^{1}\;z_{1}^{2} \cdots z_{1}^{q}}\;
        \highlight{z_{2}^{1}\;z_{2}^{2} \cdots z_{2}^{q}}
        \cdots
        \highlight{z_{m'}^{1}\;z_{m'}^{2} \cdots z_{m'}^{q}} \\
        {}&\highlight{c_{1}\;c_{2} \cdots c_{k'}}.
    \end{align*}
    For each edge $e_{j} = v_{a}v_{b} \in E(G)$ with $a < b$, we introduce a permutation $\pi_{e_{j}}$ obtained from $\pi_{\base}$ as follows.
    Suppose $v_{a}$ has color $c_{\alpha}$ and $v_{b}$ has color $c_{\beta}$.
    Swap $x_{j}^{i}$ and $y_{j}^{i}$ for each $i\in [q]$; move $c_{\alpha}$ to the right of $v_{a}^{q}$; move $c_{\beta}$ to the right of $v_{b}^{q}$; and move the remaining color-symbols to the right of $z^{q}_{j}$ without changing their relative order.
    That is,
    \begin{align*}
        \pi_{e_{j}} \coloneq {}&
        \highlight{x_{1}^{1}\;y_{1}^{1}\cdots x_{1}^{q}\;y_{1}^{q}}
        \cdots
        \HIGHlight{y_{j}^{1}\;x_{j}^{1} \cdots y_{j}^{q}\;x_{j}^{q}}
        \cdots
        \highlight{x_{m'}^{1}\;y_{m'}^{1} \cdots x_{m'}^{q}\;y_{m'}^{q}}\\
        {}&\highlight{v_{1}^{1} \cdots v_{1}^{q}}
        \cdots
        \highlight{v_{a}^{1} \cdots v_{a}^{q}}\;
        \HIGHlight{c_{\alpha}}
        \cdots
        \highlight{v_{b}^{1} \cdots v_{b}^{q}}\;
        \HIGHlight{c_{\beta}}
        \cdots
        \highlight{v_{n'}^{1} \cdots v_{n'}^{q}}\\
        {}&\highlight{z_{1}^{1} \cdots z_{1}^{q}}
        \cdots
        \highlight{z_{j}^{1} \cdots z_{j}^{q}}\;
        \HIGHlight{\chi_{\alpha,\beta}}\;
        \highlight{z_{j+1}^{1} \cdots z_{j+1}^{q}}
        \cdots
        \highlight{z_{m'}^{1} \cdots z^{q}_{m'}},
    \end{align*}
    where $\chi_{\alpha,\beta} = \highlight{c_{1} \cdots c_{\alpha-1}\;c_{\alpha+1} \cdots c_{\beta-1}\;c_{\beta+1} \cdots c_{k'}}$.
    Finally, recall that $k = m' - \binom{k'}{2} + 1$.
    We set $d = \binom{k'}{2}(q + (k'- 2))$, and let the set of permutations consist of all permutations for edges in the graph, i.e.,
    $\Pi = \{\pi_{e_{j}} : e_{j} \in E(G)\}$.

    \begin{lemma}\label{lem:w1firstdirection}
        The stated reduction from \textsc{MultiColoredClique} to \umedianfull preserves \yesinstances.
    \end{lemma}
    \begin{proof}
        Assume that $G$ contains a clique $Q$ of size at least $k'$, which contains exactly one vertex of each color.
    We prove that there is a set $S$ of size $k$ that is a solution for the \umedianfull instance $(\Pi, k, d)$.
    For each vertex $v_{j} \in Q$, let $c_{\varphi(j)}$ be the color of $v_{j}$.
    Define $\sigma_{Q}$ to be the permutation obtained from $\pi_{\base}$ by moving each symbol $c_{\varphi(j)} \in C$ right after the symbol $v_{j}^{q} \in Q$.
    (There is no color-symbol after $v_{j}^{q}$ if $v_{j} \not\in Q$.)
    That is, if $Q = \{v_{j_{1}}, v_{j_{2}}, \ldots, v_{j_{k'}}\}$ with $j_{1} < j_{2} < \cdots < j_{k'}$, then
    \begin{align*}
        \sigma_{Q} \coloneq 
        {}& \highlight{x_{1}^{1}\;y_{1}^{1} \cdots x_{1}^{q}\;y_{1}^{q}}\;
        \highlight{x_{2}^{1}\;y_{2}^{1} \cdots x_{2}^{q}\;y_{2}^{q}}
        \cdots
        \highlight{x_{m'}^{1}\;y_{m'}^{1} \cdots x_{m'}^{q}\;y_{m'}^{q}}\;\\
        {}&\highlight{v_{1}^{1} \cdots v_{1}^{q}}
        \cdots 
        \highlight{v_{j_{1} - 1}^{1} \cdots v_{j_{1} - 1}^{q}} \;
        \highlight{v_{j_{1}}^{1} \cdots v_{j_{1}}^{q}}\;\HIGHlight{c_{\varphi(j_{1})}} \;
        \highlight{v_{j_{1} + 1}^{1} \cdots v_{j_{1} + 1}^{q}} \\
        {}&\phantom{\highlight{v_{1}^{1} \cdots v_{1}^{q}}}
        \cdots
        \highlight{v_{j_{i} - 1}^{1} \cdots v_{j_{i} - 1}^{q}} \;
        \highlight{v_{j_{i}}^{1} \cdots v_{j_{i}}^{q}}\;\HIGHlight{c_{\varphi(j_{i})}} \;
        \highlight{v_{j_{i} + 1}^{1} \cdots v_{j_{i} + 1}^{q}} \\
        {}&\phantom{\highlight{v_{1}^{1} \cdots v_{1}^{q}}}
        \cdots
        \highlight{v_{j_{k'} - 1}^{1} \cdots v_{j_{k'} - 1}^{q}} \;
        \highlight{v_{j_{k'}}^{1} \cdots v_{j_{k'}}^{q}}\;\HIGHlight{c_{\varphi(j_{k'})}} \;
        \highlight{v_{j_{k'} + 1}^{1} \cdots v_{j_{k'} + 1}^{q}}\\
        {}& \cdots
        \highlight{v_{n'}^{1} \cdots v_{n'}^{q}}\; 
        \highlight{z_{1}^{1} \cdots z_{1}^{q}}\;
        \highlight{z_{2}^{1} \cdots z_{2}^{q}}
        \cdots
        \highlight{z_{m'}^{1} \cdots z_{m'}^{q}}.
    \end{align*}

    We then let $S = \{\sigma_{Q}\} \cup \Pi \setminus \Pi_{Q}$, where $\Pi_{Q} \subseteq \Pi$ is the set of $\binom{k'}{2}$ permutations associated with an edge that has both endpoints covered by $Q$.
    Observe that $|S| = m' - \binom{k'}{2} + 1 = k$ and the sum of distances of each permutation to its closest median is at most $\sum_{\pi \in \Pi_{Q}} \dist(\pi, \sigma_{Q})$.
    Any permutation in $\Pi_{Q}$ can be transformed into $\sigma_{Q}$ by moving the $k'-2$ color-symbols that are not yet at the correct position as well as the $q$ symbols of the form $x_{j}^{i}$ with $i \in [q]$.
    Thus, the total sum of distances is at most $\binom{k'}{2}(q + (k'- 2)) = d$.
    \end{proof}

 \begin{lemma}\label{lem:w1seconddirection}
    The stated reduction from \textsc{MultiColoredClique} to \umedianfull preserves \noinstances.
\end{lemma}
\begin{proof}
    We prove by showing that the existence of a solution to the constructed \umedianfull instance implies the existence of a solution to the \textsc{MultiColoredClique} instance.

    Assume the \umedianfull instance is solved by a set $S$ of permutation medians which minimize the sum of distances.
    Consider the clustering of $\Pi$ described by $S$, where we have one cluster per median permutation in $S$ and each permutation in $\Pi$ is assigned to the cluster of the closest median under Ulam metric (breaking ties arbitrarily).
    For ease of presentation, we partition the alphabet $\Sigma$ into two parts: $\Sigma'$ is the set of $x$- and $y$-symbols; and $\Sigma''$ is the set of $v$-, $z$-, and color-symbols.
    By construction, any permutation $\pi_{e}$ in $\Pi$ is a concatenation of two substrings $\pi'_{e} \in \Scal_{2m'q}(\Sigma')$ and $\pi''_{e} \in \Scal_{k' + n'q + m'q}(\Sigma'')$.
    It is easy to see that $\dist(\pi_{e_{1}}, \pi_{e_{2}}) = \dist(\pi'_{e_{1}}, \pi'_{e_{2}}) + \dist(\pi''_{e_{1}}, \pi''_{e_{2}})$ for any pair of permutations $\pi_{e_{1}}, \pi_{e_{2}} \in \Pi$.
    We further define $C' = \{ \pi'_{e} : \pi_{e} \in C \}$ and $C'' = \{ \pi''_{e} : \pi_{e} \in C \}$ for any cluster $C$ in the clustering.

    Clearly, there exists a cluster of size at least two as $k < m'$.
    Consider a cluster $C$ of size $|C| \geq 2$. 
    We first argue that $C$ is the unique cluster of size more than one.
    Let $\sigma$ be its median with optimal total distance to the permutations in $C$.
    By the optimality of $\sigma$, the center $\sigma$ starts with a substring $\sigma' \in \Scal_{2m'q}(\Sigma')$ and end with a substring $\sigma'' \in \Scal_{k' + n'q + m'q}(\Sigma'')$.
    It is easy to see that $\dist(\pi_{e}, \sigma) = \dist(\pi'_{e}, \sigma') + \dist(\pi''_{e}, \sigma'')$ for any permutation $\pi_{e} \in C$.
    This means that we can optimize $\sigma'$ and $\sigma''$ independently to minimize the total distance.
    First, restrict our attention to the $x$- and $y$-symbols.
    For any pair of permutations $\pi'_{e_{1}}, \pi'_{e_{2}} \in C'$, we have $\dist(\pi'_{e_{1}}, \pi'_{e_{2}}) \geq 2q$ as they have swapped the $x$- and $y$-symbols for distinct edges.
    Thus, we can derive that
    \begin{align*}
        \sum_{\pi'_{e} \in C'} \dist(\pi'_{e}, \sigma')
        ={}& \frac{1}{|C'| - 1} \sum_{\{\pi'_{e_{1}}, \pi'_{e_{2}}\} \subseteq C'} (\dist(\pi'_{e_{1}}, \sigma') + \dist(\pi'_{e_{2}}, \sigma')) \\
        \geq{}& \frac{1}{|C'| - 1} \sum_{\{\pi'_{e_{1}}, \pi'_{e_{2}}\} \subseteq C'} \dist(\pi'_{e_{1}}, \pi'_{e_{2}}) \\
        \geq{}& \frac{1}{|C'| - 1} \cdot \binom{|C'|}{2} \cdot 2q = |C'| q = |C|q.
    \end{align*}
    Suppose for the contraction that there were $r$ clusters $C_{1}, C_{2}, \ldots, C_{r}$ ($r \geq 2$) of size at least two with respective medians $\sigma_1', \sigma_2', \ldots \sigma_r'$.
    Summing over all these clusters yields 
    \begin{equation*}
        (|C'_{1}| + |C'_{2}| + \cdots + |C'_{r}|)q 
        \leq \sum_{i \in [r]}\sum_{\pi'_{e} \in C'_{i}} \dist(\pi'_{e}, \sigma_i')
        \leq d 
        = \binom{k'}{2}\big(q + (k'- 2)\big) < \Big(\binom{k'}{2} + 1 \Big)q,
    \end{equation*}
    which implies that $|C_{1}| + |C_{2}| + \cdots + |C_{r}|\leq \binom{k'}{2}$.
    By the definition of $k$ we have
    \begin{align*}
        m' + r - (|C_{1}| + |C_{2}| + \cdots + |C_{r}|) \leq k = m' - \binom{k'}{2} + 1,
    \end{align*}
    leading to $r \leq 1$, a contradiction.
    Hence, we conclude that $C$ is the unique cluster of size more than one, which indicates $|C| \geq m' - k' + 1 = \binom{k'}{2}$.

    Since $(\Pi, k, d)$ is a \yesinstance, we have $d$ is at least
    \begin{equation*}
        \sum_{\pi_{e} \in C} \dist(\pi_{e}, \sigma) 
        = \sum_{\pi'_{e} \in C'} \dist(\pi'_{e}, \sigma') + \sum_{\pi''_{e} \in C''} \dist(\pi''_{e}, \sigma'')
        \geq |C'| q + \sum_{\pi''_{e} \in C''} \dist(\pi''_{e}, \sigma'').    
    \end{equation*}
    Since $d = \binom{k'}{2} (q + (k'-2))$ and $|C'| = |C| \geq \binom{k'}{2}$, we derive that 
    \begin{equation*}
        \sum_{\pi''_{e} \in C''} \dist(\pi''_{e}, \sigma'') \leq \binom{k'}{2}(k'-2) = q - 1.
    \end{equation*}
    This implies that only color-symbols can be moved from $\sigma''$ to each permutation $\pi''_{e}$ in $C''$: there are $q$ copies of each $v$- and $z$-symbol in $\Sigma''$ and they are ordered the same way in all permutations.
    Suppose $\bar{\Ccal} \subseteq \Ccal$ is the set of colors located before $z_{1}^{1}$ in $\sigma''$, and let $\bar{k} = |\bar{\Ccal}|$.
    Recall that each permutation $\pi_{e}$ has two of the $k' \geq 4$ color-symbols placed before $z_{1}^{1}$; one each after the symbols for its incident vertices.
    Thus, for at most $\binom{\bar{k}}{2}$ of the permutations in $C''$, both of these color-symbols are in $\bar{\Ccal}$.
    All other permutations have to move at least one of the two symbols, which requires a total of at least $\binom{k'}{2}-\binom{\bar{k}}{2}$ units of the budget.
    Additionally, for each permutation in $C''$, all $k'-2$ color-symbols after $z_{1}^{1}$ have to be moved, except a total of at most $k'-\bar{k}$ symbols across all permutations.
    As a result, the total number of required move operations is at least
    \begin{align*}
        \sum_{\pi''_{e} \in C''} \dist(\pi''_{e}, \sigma'')
        \geq{} &\binom{k'}{2}-\binom{\bar{k}}{2} + \binom{k'}{2}(k'-2) - \big(k'-\bar{k}\big)\\
        ={}& \binom{k'}{2}(k'-2) + \frac{1}{2} \Big(k'(k'-3) - \bar{k}\big(\bar{k} - 3 \big)\Big)
        \geq \binom{k'}{2}(k'-2).
    \end{align*}
    Combine this with the bound on $\sum_{\pi'_{e} \in C'} \dist(\pi'_{e}, \sigma')$, and we finally have
    \begin{align*}
        \sum_{\pi_{e} \in C} \dist(\pi_{e}, \sigma)
        ={}& \sum_{\pi'_{e} \in C'} \dist(\pi'_{e}, \sigma') + \sum_{\pi''_{e} \in C''} \dist(\pi''_{e}, \sigma'')\\
        \geq{}& \binom{k'}{2} q + \binom{k'}{2}(k'-2) = d.
    \end{align*}
    Hence, all inequalities above are in fact equalities, which yields $|C'| = k'$ and $\bar{k} = k'$.
    In other words, all color-symbols are placed somewhere before $z_{1}^{1}$ in $\sigma$.
    Note further that the budget $d$ is only sufficient if for every permutation $\pi_{e} \in C$ both color-symbols ``incident'' to $e$ are located between the same substrings of $v$-symbols as in $\sigma$.
    Hence, the union of all endpoints of the $\binom{k'}{2}$ distinct edges described by $C$ has size $|C'| = k'$, implying that these edges form a $k'$-clique in $G$.
\end{proof}

The $\W[1]$-hardness of \umedianfull immediately follows from \cref{lem:w1firstdirection,lem:w1seconddirection} and the facts that the described reduction can be computed in polynomial time and determines $d$ by a computable function of $k'$.
\medianWhard*

We contrast this lower bound by designing a simple $\XP$-algorithm for \umedianfull.

\medianXP*
\begin{proof}
    A \yesinstance of \umedianfull can be characterized by a sequence of at most $d$ moves (deletion and insertion) on the input set of permutations $\Pi$.
    If after performing at most $d$ such moves, we obtain at most $k$ distinct permutations then we can say these $k$ permutations are the $k$ medians.
    From the input set $\Pi$, we have $mn^2$ possible choices of single moves: choose a permutation from $\Pi$, choose an element to delete (there are $n$ choices), and choose a position to reinsert it (again $n$ choices).
    As a result, trying all combinations of at most $d$ moves takes time in $\Ocal\big((mn)^{2d}\big)$.
\end{proof}

\subsection{Fixed-Parameter Tractability by \texorpdfstring{$k+d$}{k+d}}

In this subsection, we show that \umedianfull admits a polynomial kernel parameterized by both $k$ and $d$.
Compared to \ucenterfull, this matches the fixed-parameter tractability by $d + k$ but contrasts our lower bound on polynomial kernels.

Our kernelization algorithm is based on three reduction rules.
A reduction rule is \emph{safe} if it transforms an instance into an equivalent instance (preserving \yesinstances and \noinstances).
When introducing a reduction rule, we assume all previous reduction rules cannot be applied.

\begin{reduction} \label{R1}
    If there are more than $k+d$ distinct permutations in $\Pi$, reduce to a trivial \noinstance of constant size.
\end{reduction}

\begin{reduction} \label{R2}
    If a permutation $\pi \in \Pi$ occurs more than $d + 1$ times in $k$, then delete all but $d+1$ occurrences of $\pi$ from $\Pi$.
\end{reduction}

\begin{lemma}
    Reduction rule~\ref{R1} and Reduction rule~\ref{R2} are safe.
\end{lemma}

\begin{proof}
    Note that each move-operation reduces the number of distinct permutations in $\Pi$ by at most one.
    If there are more than $k + d$ distinct permutations, it is impossible to reduce the number of distinct permutations to at most $k$ using at most $d$ move-operations.
    This shows the correctness of rule \ref{R1}.

    Now, we prove the correctness of rule \ref{R2}.
    If the permutation $\pi \in \Pi$ occurs more than $d$ times in $\Pi$, it has to be a median in the solution; otherwise, at least $d + 1$ move-operations are needed to transform all these occurrences to a median distinct from $\pi$.
    This implies that keeping $d + 1$ occurrences of $\pi$ is sufficient to obtain an equivalent instance.
\end{proof}

After exhaustively applying \ref{R1} and \ref{R2}, there are at most $k + d$ distinct permutations in $\Pi$, and each permutation occurs at most $d + 1$ times.
Then we apply the last reduction rule, \ref{R3}, to bound the length of the permutations.
This reduction rule is based on the concepts of a \emph{distance-preserving} translation (defined in Section~\ref{sec:nokernel}) and \emph{maximal connectivity components}.

Recall that the previously defined function $f\colon \{0, 1\}^{\ell} \rightarrow \Scal_{2\ell}([2\ell])$ is distance-preserving as it replaces the $i$-th bit of a bitstring by the pair of symbols $(2i - 1, 2i)$ if the bit is $0$ and by the pair $(2i, 2i - 1)$ if the bit is $1$.
The definition of connectivity components is as follows.

\begin{definition}[Connectivity component]
    Let $\Pi \subseteq \Scal_n(\Sigma)$ be a (multi)set of permutations and $d$ be an integer.
    A sequence of permutations $\pi_{1}, \pi_{2}, \ldots, \pi_{\ell}$ in $\Pi$ is called a \emph{path} if every neighboring pair $\pi_{i}, \pi_{i+1}~(i \in [r - 1])$ has Ulam distance at most $d$.
    A \emph{connectivity component} $C \subseteq \Pi$ is a subset of permutations such that every pair $\pi_{1}, \pi_{2}$ of permutations in $C$ has a path in $C$ with endpoints $\pi_{1}$ and $\pi_{2}$.
    A connectivity component is maximal if it is not a proper subset of another connectivity component.
\end{definition}

According to the definition, for a partition of $\Pi$ into maximal connectivity components, every pair of permutations in distinct components has Ulam distance more than $d$.
It is easy to partition $\Pi$ into maximal connectivity components $(C_{1}, C_{2}, \ldots, C_{r})$ via a greedy algorithm that exhaustively performs breadth-first searches (or depth-first searches) from uncovered permutations. 

\begin{reduction} \label{R3}
    Consider any maximal common substring $s_{i}$ shared by all permutations in a maximal connectivity component $C_i$.
    If $s_i$ consists of more than $d + 1$ symbols, replace it by a sequence of $d+1$ new symbols in all permutations in $C_i$.
    Then create a new instance using a modified alphabet and the modified set of permutations as follows \textup{(}cf. \cref{fig:exampleReductionRule}\textup{)}:
    \begin{enumerate}
        \item Relabel all symbols to be consecutive positive integers starting at $2(d + 1)r + 1$;
        \item Using the symbols in $[2(d+1)r]$, add the prefix $f\big(0^{(d+1)(i-1)}\;1^{d+1}\;0^{(d+1)(r - i)} \big)$\footnote{Here we let $0^\ell$ and $1^\ell$ denote the bitstring of $\ell$ consecutive zeros or ones, respectively.} to every permutation in each $C_{i}$ with $i \in [r]$, where $f$ is the distance-preserving translation defined above.
        \item Let $\Sigma'$ be the set of symbols in the longest permutation.
        To each permutation in $\Pi$, add the missing symbols from $\Sigma'$ in increasing order as a suffix. 
    \end{enumerate}
\end{reduction}

We call an instance of \umedianfull \emph{reduced} if none of the reduction rules \ref{R1}-\ref{R3} is applicable.

\begin{lemma}
    Reduction rule \ref{R3} is safe.
\end{lemma}

\begin{proof}
    We prove that a reduced instance $\Ical'$ with permutations $\Pi'$ and alphabet $\Sigma'$ is a \yesinstance if and only if so is the original instance $\Ical$.
    For both directions, observe that if all permutations covered by some median share a substring of length more than $d$, then without loss of generality, that median contains the same substring as well.

    First, assume there is a set $S$ of permutations over $\Sigma$ witnessing $\Ical$ to be a \yesinstance and associate each permutation in $\Pi$ with the closest median in $S$ (breaking ties arbitrarily).
    Observe that all permutations associated to the same median belong to the same connectivity component.
    Let $S'$ be the set of permutations over $\Sigma'$ obtained by adding one median $\sigma'$ for each median $\sigma\in S$ with at least one associated permutation as follows.
    Let the associated permutation(s) be in the $i$-th connectivity component $C_i$.
    Note that we can assume that if all permutations in $C_i$ share a substring of length at least $d+1$, then $\sigma$ contains this substring as well. We can thus apply the same contraction of substrings and relabeling to $\sigma$ as we did for the permutations in $C_i$ in the construction of the new instance. We then obtain $\sigma'$ by adding the prefix and suffix of connectivity component $C_i$.
    This direction of the proof is concluded by the fact that now each permutation in $C_i$ has the same Ulam distance to $\sigma$ as its modified permutation has to $\sigma'$.

    For the other direction, assume there is a set $S'$ of permutations over $\Sigma'$ witnessing $\Ical'$ to be a \yesinstance and associate each permutation in $\Ical'$ with the closest median in $S'$ (breaking ties arbitrarily).
    Note that pairs of permutations that originate from distinct connectivity components have Ulam distance more than $d$ by their prefix alone and are thus never associated with the same median.
    Consequently, all permutations associated with the same median have the same prefix and suffix.
    For each median $\sigma'\in S'$ associated with at least one permutation, drop the prefix and suffix and revert the relabeling and the contraction of substrings as applied for the respective connectivity component.
    This yields a permutation $\sigma$ of $\Sigma$ such that, for each modified permutation $\pi'$ associated with $\sigma'$, the Ulam distance between $\sigma'$ and $\pi'$ is the same as the Ulam distance between $\sigma$ and the respective original permutation in $\Pi$.
\end{proof}

\newcommand{\circnum}[1]{\raisebox{0pt}[.3ex][.3ex]{\text{\smash{\textcircled{\scriptsize #1}}}}
\vphantom{\rule{0pt}{\heightof{l}}}
}
\newcommand{\boxnum}[1]{%
    \begingroup
        \setlength{\fboxsep}{.7pt}%
        \setlength{\fboxrule}{.2pt}%
        \mathbin{\raisebox{-.5pt}{\fbox{\scriptsize #1}}}%
    \endgroup}

\begin{figure}
    \centering
    \begin{tikzpicture}
        \node[fill=gray!20, draw=gray!50, rounded corners=2pt, inner sep=4pt, text width=3.8cm,style={align=center}]
            (con1) at (0,0)
            { \small$\text{A B C}\, \highlight{\text{D E} \cdots \text{K}}\;\highlight{\text{L} \cdots \text{Z}}$ \\
                $\text{C A B}\, \highlight{\text{D E} \cdots \text{K}}\;\highlight{\text{L} \cdots \text{Z}}$ \\
                $\text{A B}\, \highlight{\text{D E} \cdots \text{K}}\;\text{C}\;\highlight{\text{L} \cdots \text{Z}}$ \\};
        \node[fill=gray!20, draw=gray!50, rounded corners=2pt, inner sep=4pt, text width=3.8cm,style={align=center}]
            (con2) at (0, -1.8)
            { \small$\text{Z}\, \highlight{\text{Y X W V U} \cdots \text{C B A}\,}$ \\
                $\highlight{\,\text{Y X W V U} \cdots \text{C B A}}\, \text{Z}$ \\};
      
        \node[inner sep=0pt, text width=3.8cm, style={align=center},]
            (con1) at (4.7,0)
            {\small$\text{A B C}\;\highlight{\circnum{1} \cdots \circnum{4}}\;\highlight{\boxnum{1} \cdots \boxnum{4}}$ \\
                $\text{C A B}\;\highlight{\circnum{1} \cdots \circnum{4}}\;\highlight{\boxnum{1} \cdots \boxnum{4}}$ \\
                $\text{A B}\;\highlight{\circnum{1} \cdots \circnum{4}}\;\text{C}\;\highlight{\boxnum{1} \cdots \boxnum{4}}$ \\};

        \node[inner sep=0pt, text width=5cm, style={align=center},]
            (pre1) at (9.2,0)
            {\small $\highlight{2\;1} \cdots \highlight{8\;7}\;\highlight{9\;10} \cdots \highlight{15\;16}$};
            
        \node[inner sep=0pt, text width=5cm, style={align=center},]
            (pre2) at (9.2,-1.8)
            {\small $\highlight{1\;2} \cdots \highlight{7\;8}\;\highlight{10\;9} \cdots \highlight{16\;15}$};

        \node[inner sep=0pt, text width=5cm, style={align=center},]
            (con2) at (4.5,-1.8)
            {\small$ \text{Z}\;\highlight{\circnum{1} \cdots \circnum{4}}$\\ 
                $\highlight{\circnum{1} \cdots \circnum{4}}\;\text{Z}$};

        \draw[-{Straight Barb[length=.8mm]}] (2.2,0) -- (2.7,0);  
        \draw[-{Straight Barb[length=.8mm]}] (2.2,-1.8) -- (2.7,-1.8);  
             
        \node[inner sep=0pt, style={align=center}] (label1) at (0,1.5) {$r=2$ connectivity\\ components of $\Pi$};
        \node[inner sep=0pt, style={align=center}] (label2) at (4.7,1.5) {Contraction of\\ common substrings};
        \node[inner sep=0pt, style={align=center}] (label3) at (9,1.5) {Prefix};
     
        \node[inner sep=0pt, style={align=center}] (fulltranslation) at (4.9,-4.2) {
        \small $\text{A B C}\, \highlight{\text{D E} \cdots \text{K}}\;\highlight{\text{L} \cdots \text{Z}} \rightarrow \highlight{2\;1} \cdots \highlight{8\;7}\;\highlight{9\;10} \cdots \highlight{15\;16}~ 17\;18\;19\;\highlight{20\;21\;22\;23}\;\highlight{24\;25\;26\;27}$ \\
        \small $\text{C A B}\, \highlight{\text{D E} \cdots \text{K}}\;\highlight{\text{L} \cdots \text{Z}} \rightarrow \highlight{2\;1} \cdots \highlight{8\;7}\;\highlight{9\;10} \cdots \highlight{15\;16}~ 19\;17\;18\;\highlight{20\;21\;22\;23}\;\highlight{24\;25\;26\;27}$ \\
        \small $\text{A B}\, \highlight{\text{D E} \cdots \text{K}}\;\text{C}\;\highlight{\text{L} \cdots \text{Z}} \rightarrow \highlight{2\;1} \cdots \highlight{8\;7}\;\highlight{9\;10} \cdots \highlight{15\;16}~ 17\;18\;\highlight{20\;21\;22\;23}\;19\;\highlight{24\;25\;26\;27}$ \\
        \\
        \small $\text{Z}\, \highlight{\text{Y X W V U} \cdots \text{C B A}\,} \rightarrow \highlight{1\;2} \cdots \highlight{7\;8}\;\highlight{10\;9} \cdots \highlight{16\;15}~ 17\;\highlight{18\;19\;20\;21}\;\highlight{22\;23\;24\;25\;26\;27}$ \\
        \small $\highlight{\,\text{Y X W V U} \cdots \text{C B A}}\, \text{Z} \rightarrow \highlight{1\;2} \cdots \highlight{7\;8}\;\highlight{10\;9} \cdots \highlight{16\;15}~ \highlight{18\;19\;20\;21}\;17\;\highlight{22\;23\;24\;25\;26\;27}$
        };
             
        \end{tikzpicture}
        \caption{Application of reduction rule \ref{R3} for an instance with $d=3$. In the top left, five input permutations are partitioned into two maximal connectivity components. Common substrings within the components are replaced by new substrings of length $d+1=4$:
        In the first component,
        $\highlight{\text{D E}\cdots \text{K}}$ becomes 
        $\highlight{\circnum{1}\cdots\circnum{4}}$
        and 
        $\highlight{\text{L}\cdots \text{Z}}$ becomes 
        $\highlight{\boxnum{1}\cdots\boxnum{4}}$.
        In the second component,
        $\highlight{\text{Y X W}\cdots \text{A}}$ becomes 
        $\highlight{\circnum{1}\cdots\circnum{4}}$.
        For each component, all symbols are relabeled to integers starting at $2(d+1)r+1=17$ (i.e., A, B, C, $\circnum{1}, \dots$ respectively become $17$, $18$, $19$, $20, \dots$). 
        Each permutation obtains a prefix of the symbols $1$ to $16$ in increasing order, except that for the first (second) connectivity component the first (second) set of $d+1=4$ pairs of symbols are swapped.
        The last lines of the figure depict the transformed permutations. There, permutations of the second connectivity component have a suffix of the symbols $22$ to $27$, these were used in the first component ($11$ distinct symbols after contraction) but not required in the second one ($5$ distinct symbols after contraction).
        } 
        \label{fig:exampleReductionRule}
    \end{figure}
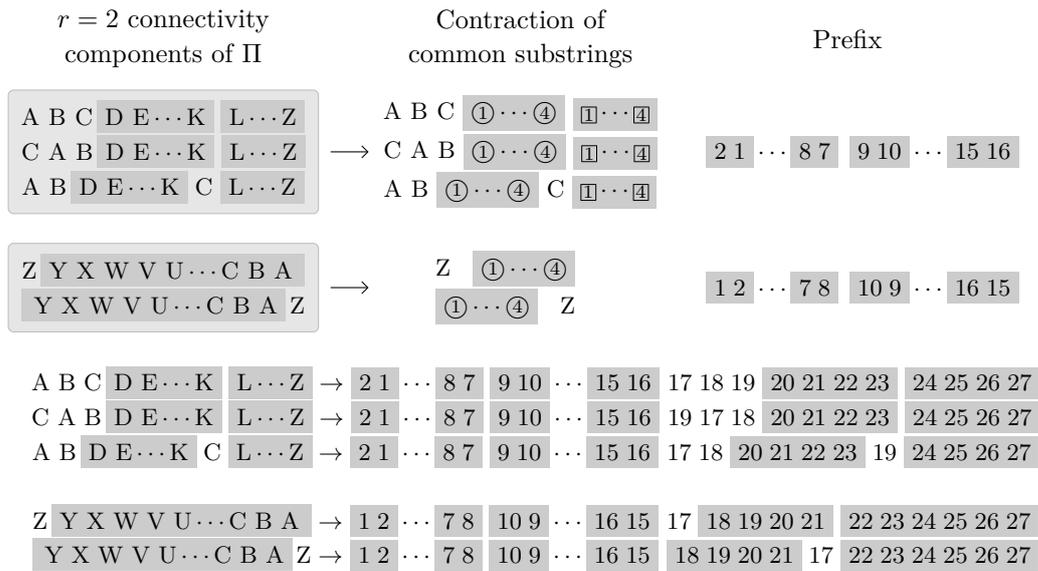

Finally, we are ready to prove that \umedianfull admits a polynomial kernel parameterized by $k + d$.
    
\mediankernel*

\begin{proof}
   Exhaustively applying rules \ref{R1} and \ref{R2} takes time in $\Ocal(m^2n)$ and reduces the number of permutations to $(d+k)(d+1)$.
   Rule \ref{R3} can be applied in time in $\Ocal(m^2 n \log n)$, as computing the Ulam distance between a pair of permutations (by finding a longest common subsequence) takes time in $\Ocal(n\log n)$.
   To obtain the stated kernel size, we now bound the size of the alphabet after the application of \ref{R3} to $\Ocal(d^2(d+k)^2)$.
   For each $i\in [r]$, note that there are at most $d+k$ distinct permutations in $C_i$ and so each permutation $\pi$ in $C_i$ has Ulam distance at most $d(d+k-1)$ to each other permutation in $C_i$.
   Fix any permutation $\pi\in C_i$. Consider a minimum set of move operations that can be applied to the set of distinct permutations in $C_i$ to make them equal to $\pi$.
   This set contains at most $d' = (d+k-1)^2d$ operations. In $\pi$, mark each symbol red that is moved at least once, and place a red separator between each pair of consecutive symbols where at least one symbol is inserted.
   Then all substrings of $\pi$ that do not contain any red symbol or separator are shared substrings of all permutations in $C_i$. 
   These form at most $2d'+1$ maximal shared substrings, as the total number of red symbols and red separators is at most $2d'$. 
   Recall that each substring is replaced by at most $(d+1)$ symbols. Including the at most $d'$ red symbols and the length of the prefix gives that the number of symbols in the modified alphabet is at most
   \begin{align*}
       |\Sigma'| 
       \le (d+1)(2d'+1) + d' + 2(d+1)(d+k)
       = \Ocal(d^2 (d+k)^2)
       = \Ocal(d^4 + d^2k^2). \tag*\qedhere
   \end{align*}  
\end{proof}
    
  \section{Conclusion}\label{sec:conclusion}
    
We initiated the study of the parameterized complexity of the \ucenterfull and
\umedianfull problems.
Prior to our work, even the parameterized complexity for
the case $k = 1$ with parameter $d$ was unknown for these problems. 

It is known~\cite{Cygan2015parameterized} that \cstringfull cannot be solved in $2^{o(d\log d)}\cdot (nm)^{\Ocal(1)}$ time unless the Exponential Time Hypothesis (ETH) fails (we refer to~\cite{Cygan2015parameterized} for the formal definition of ETH).
Then the Polynomial Parameter Transformation from \Cref{lem:red_string_to_center} implies the same computational lower bound for \ucenterone.
Is it possible to solve \ucenterone for $k=1$ in $d^{\Ocal(d)}\cdot(nm)^{\Ocal(1)}$ time similarly to \cstringfull?

For \umedianfull, it is shown that the problem is $\W[1]$-hard when parameterized by~$d$ but is in \XP{} under the same parameterization.
Furthermore, the problem admits a polynomial kernel for the parameterization by both $k$ and $d$.
The latter result together with an \XP-algorithm implies that \umedianfull can be solved in $(kd)^{\Ocal(d)}\cdot (mn)^{\Ocal(1)}$ time. 
Is it possible to improve this running time? Similarly to \ucenterfull, we find that this question is interesting even for $k=1$. Is there a single-exponential (or even subexponential) in $d$ algorithm?

\clearpage



\bibliography{refs}

\end{document}